\documentclass[]{article}
\usepackage[T1]{fontenc}
\usepackage{etoolbox}
\usepackage{cite}
\usepackage{amsmath,amssymb,amsfonts, amsthm}
\usepackage{algorithm, algpseudocodex}
\usepackage{graphicx}
\usepackage{textcomp}
\usepackage{xcolor}
\usepackage{xspace}
\usepackage[colorinlistoftodos,prependcaption,textsize=footnotesize]{todonotes}
\usepackage{hyperref}
\usepackage{numprint} 
\usepackage{subcaption}
\newcommand{\DS}[1]{\todo[author=Dominik,inline,color=green!50]{ #1}\xspace}

\newcommand{\budget}{\mathcal{G}\xspace}
\newcommand{\proc}{{\mu}\xspace}
\newcommand{\PPP}{\mathcal{P}\xspace}

\newcommand{\idleP}[1]{\mathcal{P}_{\text{idle}}^{#1}}
\newcommand{\workP}[1]{\mathcal{P}_{\text{work}}^{#1}}

\newcommand{\algo}[1]{\texttt{#1}\xspace}

\newcommand{\start}[1]{\sigma(#1)}
\newcommand{\finish}[1]{\sigma(#1) + t(#1)}

\newcommand{\CWM}{\algo{CWM}}
\newcommand{\cost}{\mathcal{C}\mathcal{C}}
\newcommand{\slack}{\texttt{H-CWS-s}}
\newcommand{\pressure}{\texttt{H-CWS-p}}
\newcommand{\heft}{\texttt{HEFT-SL}}
\newcommand{\new}[1]{\textcolor{black}{#1}}
\newcommand{\old}[1]{\textcolor{red}{#1}}

\newtheorem{theorem}{Theorem}[section]
\newtheorem{lemma}[theorem]{Lemma}

\theoremstyle{definition}

\theoremstyle{remark}

\numberwithin{equation}{section}

\providecommand{\keywords}[1]{\vspace{0.5em}\noindent\textbf{Keywords:} #1}
\newtoggle{TR}
\toggletrue{TR} 
\newtoggle{showOld} 
\togglefalse{showOld}
\newtoggle{showBudget} 
\togglefalse{showBudget}

\begin{document}
%
%
%
\title{Carbon-Aware Mapping and Scheduling for Deadline-Constrained Workflows}
%
\title{Carbon-Aware Mapping and Scheduling for Deadline-Constrained Workflows}
\date{}
\author{
  Dominik Schweisgut\thanks{Karlsruhe Institute of Technology (KIT), Germany, \texttt{dominik.schweisgut@kit.edu}}
  \and
  Anne Benoit\thanks{ENS Lyon, France, \texttt{Anne.Benoit@ens-lyon.fr}}
  \and
  Yves Robert\thanks{ENS Lyon, France, \texttt{Yves.Robert@ens-lyon.fr}}
  \and
  Henning Meyerhenke\thanks{Karlsruhe Institute of Technology (KIT), Germany, \texttt{meyerhenke@kit.edu}}
}
\maketitle
\begin{abstract}
As datacenters continue to grow in scale, their energy consumption and resulting carbon footprint have become pressing concerns.
\iftoggle{showOld}{\old{, where the execution of each task depends on the completion of its predecessors
With the increasing share of renewable energy, more and more datacenters operate on a mixed energy supply,
where part of the energy comes from ``green'' sources.
To reduce carbon emissions resulting from executing these workflows,
it is promising to shift task execution to periods where the availability of green power is high.}}{}
\new{With the increasing share of renewable energy in a datacenter's mixed energy supply, shifting task execution to periods of high green-power availability is a promising strategy to reduce carbon emissions.}
\new{However, in heterogeneous computing environments, 
 the power consumption of compute nodes in a datacenter can also vary.}
In practice, workloads submitted to datacenters are often not isolated tasks, but entire workflows
consisting of interdependent tasks with precedence constraints.
A further challenge arises from the fact that carbon emission reductions must typically be achieved under strict workflow deadlines. 
In this work, we show that the problem posed by these challenges for the scheduler is NP-hard and admits no constant-factor approximation 
even for the uni-processor case. 
Motivated by this hardness, we present a novel algorithm \CWM that combines carbon-aware mapping and scheduling to construct feasible solutions. 
Our approach integrates dynamic programming with efficient heuristics to exploit renewable energy availability and infrastructure heterogeneity.
To assess the quality of the new algorithm, we evaluate it against the state-of-the-art approach $\texttt{CaWoSched}$
and show that \CWM achieves significant reductions in terms of carbon emissions in experiments. 
In particular, we are able to achieve a median carbon cost reduction of $42\%$ over the best version of
$\texttt{CaWoSched}$ when the deadline is two times the makespan of a carbon-agnostic baseline.
\iftoggle{showOld}{\old{Note that \texttt{CaWoSched} in turn decreases the carbon cost of the carbon-agnostic baseline by~$36\%$,
$32\%$ or $36\%$, respectively,
exploiting the slack of the larger deadline.}}{}
\new{Note that \texttt{CaWoSched} itself already reduces the carbon-agnostic baseline by~$36\%$.}
\end{abstract}
\section{Introduction}
\label{sec.introduction}
\iftoggle{showBudget}{\DS{this should be roughly 1.5 pages}}{}
The growing awareness of environmental concerns and increasingly stringent regulations are pushing all manufacturing entities towards reducing the carbon footprint of their operations~\cite{mencaroni25}. 
Data/computing centers are also particularly affected as they represent a steadily increasing fraction of carbon emissions~\cite{AhmedBA21review}. 
In this context, data/computing centers have started to use a mix of different power sources, 
giving priority to lower carbon-emitting technologies (solar, wind, nuclear) over higher ones (coal, natural gas).
\iftoggle{showOld}{\old{
This explains that carbon-aware scheduling
has become a major technique for lowering the carbon footprint of computations and data accesses by accounting for the time-dependent
carbon intensity of the power grid and the varying availability of renewable electricity sources.}}{}
\new{This has motivated growing interest in carbon-aware scheduling as a promising technique for lowering the carbon footprint of computations and data accesses, by accounting for the time-dependent
carbon intensity of the power grid and the varying availability of renewable electricity sources.} 
While carbon-aware scheduling introduces new challenges, it also brings new opportunities for HPC scientists.
This is exemplified by this work focusing on scheduling application workflows on platforms whose power mix differs over time. 
On the one hand, this is a challenging endeavor because
designing efficient scheduling algorithms was already difficult (NP-complete  in most scenarios) 
when computers were powered by a single energy source with a constant level of carbon emissions.
On the other hand, the availability of greener time intervals is a great opportunity to revisit well-known algorithms: intuitively, one would aim at shifting tasks away from high-carbon emitting intervals,
while maintaining good overall performance.
More specifically, this work deals with scientific applications that consist of workflow graphs, where nodes represent tasks and edges represent inter-task dependencies.
Many workloads in data/computing centers can be expressed as such workflows, and the scheduling literature is rich of algorithms and complexity results.
The typical objective used to be the minimization of the total execution time (also called \emph{makespan}).
\iftoggle{showOld}{\old{In recent years, this single objective has been complemented by efforts to minimize the total energy consumption (typically using DVFS~\cite{rizvandi2012multiple}),
or to maximize the reliability (typically using replication~\cite{haque2017reliability} or checkpoints~\cite{han2018checkpointing}).
As a consequence, latest research focuses on multi-objective approaches, aiming at minimizing one metric while guaranteeing some threshold on the other metrics.
For instance one would minimize energy consumption while enforcing a makespan deadline and a reliability threshold.}}{}
\new{In recent years, objectives beyond makespan (total energy consumption~\cite{rizvandi2012multiple}, reliability~\cite{haque2017reliability}) have also been studied in multi-objective trade-offs, 
aiming at minimizing one metric while guaranteeing some threshold on other metrics, e.g., minimizing energy consumption while enforcing a makespan threshold.} 
But all these approaches rely on a stable platform whose characteristics do not vary over time.
A major novelty with regards to carbon emissions is that the platform is changing over time, 
alternating greener and browner intervals for carbon emissions, 
and thereby calling for new algorithms.
\new{A partial solution to the problem appeared in our previous work~\cite{CaWoSched}, where we focus on carbon minimization at the price of a 
restrictive assumption: that the mapping and ordering of the workflow  tasks are given (say using the well-known HEFT algorithm~\cite{HEFT}), together with a time deadline to complete the execution.}
The rule of the game is then to shift tasks across carbon intervals on each processor so as to minimize carbon emissions while maintaining the deadline. This restricted focus turns out to be sufficient to gain significant reduction of carbon emissions for a comprehensive set of experiments~\cite{CaWoSched}.
\iftoggle{showOld}{\old{
This work builds on the results of~\cite{CaWoSched}, but tackles a much more ambitious approach: rather than using a standard mapping and ordering algorithm (like HEFT), we aim at designing carbon-aware mapping and scheduling algorithms, and to assess the gains that can be expected from this holistic approach.
}}{}
\new{This work builds on the results of~\cite{CaWoSched}, but addresses a combinatorially much more difficult problem:
rather than fixing the mapping and ordering upfront (e.g., via HEFT), 
we jointly optimize mapping and scheduling, facing a much larger search space, and assess the gains that can be expected from this holistic approach.}
We stress that we target heterogeneous clusters whose processors have different speeds and different energy profiles.

\smallskip
\noindent{\bf Contributions. }
The main contributions of this paper are threefold:
(i) on the theory side, the proof that no approximation up to a constant factor exists for the simple problem instance with a unique processor (unless P=NP);
(ii) the design of sophisticated carbon-aware mapping and scheduling algorithms \new{with more flexibility and higher quality than state-of-the-art solutions such as~\cite{CaWoSched}};
and  (iii) their empirical assessment through a comprehensive campaign of simulations.
\iftoggle{showOld}{\old{The algorithms work in several phases: deadline-agnostic carbon minimization, deadline repair (whenever the original deadline is not met by the first phase) and local optimization.}}{}
Not surprisingly, the experiments demonstrate a major gain over a carbon-unaware approach.
More importantly, they further improve significantly
two state-of-the-art algorithms from~\cite{CaWoSched} that carbon-optimize a mapping produced by a HEFT variant.
%

\iftoggle{TR}{
\smallskip
\noindent{\bf Outline. }
The rest of the paper is organized as follows. Section~\ref{sec.related-work} surveys related work. In Section~\ref{sec.framework},
we detail our model and discuss complexity results.
We introduce new carbon-aware mapping and scheduling algorithms in Section~\ref{sec.algorithms} and assess their performance through an extensive set of simulations in Section~\ref{sec.experiments}.
Finally, we give concluding remarks and hints for future work in Section~\ref{sec.conclusion}.}{}
%
\section{Related Work}
\label{sec.related-work}
\iftoggle{showBudget}{\DS{Should be 1-1.25 pages}}{}
Carbon-aware computing has gained notable attention since datacenters represent a significant share of global carbon emissions and data volumes continue to rise.
\new{The current rise of AI and the associated scaling of datacenters adds massively to this effect.
Companies such as Meta are already exploring ways to run their hyperscaled datacenters carbon-neutrally~\cite{sustainableAI}.}
\iftoggle{showOld}{\old{Most likely, the carbon footprint peak of datacenters has not been achieved yet; thus, more work is needed to achieve carbon-neutral
datacenters.}}{}
Cao et al.~\cite{carbonUsageSurvey} first analyze the carbon footprint of datacenters 
and corresponding trends and then discuss ways to mitigate the trend of rising carbon footprints by datacenters.
In general, much prior work maintains a high-level workload view and does not address the scheduling of workflows (with notable exceptions described later).
For example, carbon-aware load balancing for geo-distributed datacenters shifts 
load to regions with lower carbon intensity, but considers only abstract workloads 
rather than interdependent tasks~\cite{Mahmud-Iyengar-2016}.
Likewise, cloud providers with globally distributed datacenters shift workloads in a scheduler-agnostic manner to greener sites based on renewable-energy forecasts~\cite{radovanovic2022carbon}.
However, as Wiesner et al.~\cite{wiesner2021} demonstrate, shifting tasks to periods with cleaner 
energy is also an efficient approach to save carbon emissions.
Breukelman et al.~\cite{breukelman2024carbon} 
follow a similar abstraction without interdependencies between the jobs 
and are thus incomparable to our work.
Mencaroni et al.~\cite{mencaroni25} present a carbon-aware scheduling approach for flow-shop scheduling.
Their model is formulated as a mixed-integer linear program. They solve the problem using 
a memetic algorithm with notable decrease in carbon emissions and limited effects on the makespan. 
However, due to the missing precedence constraints, the model is not comparable to our work. 
Another approach to a model without precedence constraints is the work of Wang et al.~\cite{wang25}.
They present a multi-objective scheduling approach where they use reinforcement learning to compute schedules 
with low carbon emissions and high Quality of Service. 
Moreover, Hanafy et al.~\cite{hanafy23} use the elasticity of cloud workloads by adjusting workloads
to the carbon intensity.

We move on by mentioning some general work regarding workflow scheduling. 
Since we present a plan-based (offline) approach for scheduling, 
we focus on these types of algorithms.
Note that there is also literature for online scheduling approaches. For this and more literature review, 
we refer the reader to surveys such as~\cite{liu2020survey}.
It is  well-known that plan-based scheduling is NP-hard in general~\cite{garey1979computers}.
This motivates the use of heuristics for larger problems.
There are mainly two different approaches to plan-based scheduling.
\iftoggle{showOld}{\old{First, there are partitioning-based approaches, which partition the
input workflow into blocks and assign whole blocks to processors. 
This is useful for large workflows since this reduces work compared to assigning each task individually.
Often, these approaches are further refined with other techniques to improve solution quality~\cite{ozkaya2019acyclic, ozkaya2019scalable, kulagina2024mapping,viil2018framework}.}}{}
\new{Partitioning-based approaches 
partition the input workflow into blocks and assign whole blocks to processors, trading per-task flexibility for scalability~\iftoggle{TR}{\cite{ozkaya2019scalable, kulagina2024mapping,viil2018framework}}{\cite{viil2018framework}}.} 
More closely to our work are list-based approaches. One of the most
influential examples is HEFT~\cite{HEFT}. Since we also use a variation of HEFT in this work,
we briefly mention the main idea. In the first phase, HEFT assigns priorities to all tasks.
Then, in phase two, it finds suitable assignments for the tasks in that order.
The popularity of this approach can be seen in its various extensions such as~\cite{pheft,samadi2018eheft,Shi06}.
In terms of energy and carbon-awareness, there are also adaptations of HEFT to save energy~\cite{durillo2014, durillo2015}.
However, these works  do not consider a varying carbon intensity, and hence we cannot compare
to them.
There is a noteworthy approach by Lechowicz et al.~\cite{bashir2025} that handles the time-varying 
availability of greener energy and the precedence constraints between tasks. 
In particular, they identify structurally important tasks to reduce effects on the makespan. 
However, there are several severe differences to our framework, among them the carbon cost model,
no explicit communication, and the absence of a cluster power model.

\new{The carbon-aware workflow scheduling framework $\texttt{CaWoSched}$~\cite{CaWoSched}
is the closest related work.  It} 
takes a given mapping as input, i.\,e., an assignment 
of tasks to processors, as well as the order of tasks and communications. 
The input also includes a power profile, which specifies a time horizon ending at the schedule deadline, 
partitioned into intervals, where each interval has a constant green power budget. 
Several variants of $\texttt{CaWoSched}$ 
optimize the schedule in terms of carbon cost by shifting tasks and communications to intervals with 
a high green power budget.
$\texttt{CaWoSched}$ also uses a local search approach to further improve the solution.
The main difference with this work is that the mapping and ordering are given in~\cite{CaWoSched}, 
while we aim at determining carbon-effective mappings here, \new{which requires substantially more involved algorithmic components.}
\section{Framework} 
\label{sec.framework}
\iftoggle{showBudget}{\DS{This should be roughly 1.75 pages}}{}
The goal is to map and schedule a workflow application, which is a set of interdependent tasks with precedence constraints,
onto a parallel platform with a budget of green power that varies with time, in order to minimize the carbon cost.
\iftoggle{TR}{{We first describe the application and platform, before defining mappings and schedules. Then, we state
the objective function and discuss problem complexity.}}{}

\smallskip
\noindent{\bf Application and platform.} 
The workflow application is modeled as a directed acyclic graph (DAG), where the vertices (set~$V$) represent the tasks,
and the edges (set~$E$) express the precedence constraints. The workflow is actually a weighted DAG  $G = (V, E, \omega, c)$,
where $\omega(v)$ is the amount of work required by task~$v\in V$, and $c(v_i, v_j)$ is the amount
of data to be communicated if tasks $v_i$ and $v_j$ are executed on two distinct processors (with $(v_i, v_j) \in E$). 
The goal is to map and schedule the application onto a set of $P$ heterogeneous processors interconnected
by homogeneous full-duplex communication links. For $1\leq k \leq P$, processor~$p_k$ has a speed~$s(p_k)$
and it can execute task~$v_i\in V$ within
time~$\frac{\omega(v_i)}{s(p_k)}$. 

\smallskip
\noindent{\bf Mappings and schedules. }
The mapping~$\mu$ specifies the processor $\proc(v_i)$ on which each task~$v_i\in V$ is executed.
Given the mapping, task~$v_i$ takes time $t(v_i)= \frac{\omega(v_i)}{s(\mu(v_i))}$ to execute. 
We further define a schedule~$\sigma$, which specifies the time at which each task is processed, and each communication is initiated. 
In order to neatly integrate communication costs, similarly to~\cite{CaWoSched}, 
we consider communications as additional tasks 
to be executed on communication channels. 
The task corresponding to edge  $(v_i, v_j)\in E$ is denoted by~$v_{i,j}$. 
There are $P(P-1)$ communication channels between processors, and we assume that
they can be used simultaneously. However, two communications occurring on the same channel must
be serialized. 
\iftoggle{showOld}{\old{The mapping of communication tasks onto communication channels is directly
derived from the mapping function~$\proc$.}}{}
For task~$v_{i,j}$, if $\proc(v_i)\neq \proc(v_j)$,
a communication must occur between both processors, and hence $v_{i,j}$ is mapped
onto the corresponding communication channel (between $\proc(v_i)$ and  $\proc(v_j)$), 
with a speed of~$\beta$ (the platform's bandwidth). 
The communication then takes time $t(v_{i,j}) = \frac{c(v_i, v_j)}{\beta}$
and we set $\beta=1$ for simplicity.  
Otherwise, $v_{i,j}$ is mapped onto the processor executing both $v_i$ and~$v_j$, and its
cost is then set to~$0$, since no communication is required. 

A schedule~$\sigma$ is valid if each
processor never executes more than one task at a time, and all precedence constraints are respected.
\iftoggle{showOld}{
\old{The constraints take the following form:
\begin{itemize}
\item For all $v\in V$, $\sigma(v)\geq 0$;
\item For all $v,v' \in V$ with $\proc(v)=\proc(v')$,  $\sigma(v) + t(v) \leq \sigma(v')$ or $\sigma(v') + t(v') \leq \sigma(v)$;
\item For all $(v_i, v_j)\in E$ with $\proc(v_i)=\proc(v_j)$, $\sigma(v_i) + t(v_i) \leq \sigma(v_j)$;
item For all $(v_i, v_j)\in E$ with $\proc(v_i) \neq \proc(v_j)$, $\sigma(v_i) + t(v_i) \leq \sigma(v_{i,j})$
and $\sigma(v_{i,j}) + t(v_{i,j}) \leq \sigma(v_j)$.
\end{itemize}}}{}
\new{Formally: 
(i) $\sigma(v)\geq 0$ for all $v\in V$; 
(ii) no overlaps: for all $v,v' \in V$ with $\proc(v)=\proc(v')$,  $\sigma(v) + t(v) \leq \sigma(v')$ or $\sigma(v') + t(v') \leq \sigma(v)$;
(iii) for $(v_i,v_j)\in E$ with $\proc(v_i)=\proc(v_j)$, $\sigma(v_i)+t(v_i)\leq\sigma(v_j)$;
and (iv) for $(v_i,v_j)\in E$ with $\proc(v_i)\neq\proc(v_j)$, $\sigma(v_i)+t(v_i)\leq\sigma(v_{i,j})$ and $\sigma(v_{i,j})+t(v_{i,j})\leq\sigma(v_j)$.}

\smallskip
\noindent{\bf Objective function. } 
The objective function is  to find a valid mapping and schedule in order to minimize the carbon cost,
while not exceeding a bound~$D$ on the makespan, which is the total execution time: 
$\max_{v\in V}\{\sigma(v)+ t(v)\}\leq D.$
As stated in the introduction, 
we consider that carbon intensity varies with time, since the amount of available green power depends
on many factors.  
We follow here the notation of our previous work and competitor \algo{CaWoSched}~\cite{CaWoSched}.
Hence, we
\iftoggle{showOld}{
\old{consider a time horizon $[0,T[$, divided into $J$ intervals $\{I_1, \dots, I_J\}$,
where interval $I_j$ has length $l_j$ and $\sum_{j=1}^{J} l_j = T$.
We let $I_j = [b_j, e_j[$ so that $l_j = e_j - b_j$ for every $1 \leq j \leq J$.
Within each interval~$I_j$, there is a constant green power budget~$\budget_j$ for each time unit $t\in I_j$.}}{}
\new{consider a time horizon $[0,T[$ divided into $J$ constant-budget intervals $I_j=[b_j,e_j[$, $1\leq j\leq J$, with $\sum_{j=1}^{J}(e_j-b_j)=T$ and green power budget~$\budget_j$ for each time unit $t\in I_j$.}
Note that $D\leq T$.

There are a total of $P + P(P-1) = P^2$ processors in the system, when accounting for communication channels,
and each of them consumes a different amount of static and dynamic power. Let $\PPP$ be the set of processors. 
Processor $p\in \PPP$ always consumes a static idle power~$\idleP{p}$,  
and it consumes an additional power~$\workP{p}$
when executing a task (dynamic power). 
If the power consumed by all processors 
at time~$t\in I_j$ is below~$\budget_j$, then the carbon cost is zero. 
Otherwise, the platform must use carbon-emitting power for the excess power, hence
incurring a carbon cost of $\PPP_t - \budget_j$, where $\PPP_t$ is the sum of all processor power
consumptions at time~$t$. The total carbon cost is obtained by integrating over time ($0\leq t < D$). 

\smallskip
\noindent{\bf Problem complexity. }
In~\cite{CaWoSched}, we proved that the decision problem with independent tasks (hence, no communication)
that are already mapped and ordered on a set of homogeneous processors 
(identical speeds, power consumption values $\idleP{p}=0$ and $\workP{p}=1$ for each $p\in\PPP$)
is strongly NP-complete. The problem obviously remains strongly NP-complete when we further have
to decide the mapping and scheduling. 
But the problem instance with a single processor -- while having polynomial complexity with a fixed mapping and ordering~\cite{CaWoSched} -- now becomes strongly NP-complete
due to the additional complexity of having to decide the ordering of the tasks.
Further, we find that there is no constant-factor approximation, unless P=NP
(even when the workflow DAG has no edges):
\begin{theorem}\label{thm.approx_DAG}
\new{
The decision problem of whether there is a mapping and scheduling of a set of independent tasks on a single processor
in order to minimize the carbon cost, while not exceeding a bound~$T$ on the makespan, with carbon cost not exceeding a given bound $C$,
is NP-complete in the strong sense. Furthermore, there is no polynomial-time $\lambda$-approximation algorithm 
(with $\lambda \geq 0$) such that the algorithm 
gives a solution with carbon cost $\mathcal{CC} \leq \lambda\cdot \mathcal{CC}^{*}$, 
where $\mathcal{CC}^{*}$ is the optimal carbon cost, unless $P=NP$.}
\end{theorem} 
As a direct consequence, the general problem tackled in this work
does not have any constant-factor approximation.
\new{The proof of the theorem can be found in~\iftoggle{TR}{Appendix~\ref{appendix.thmProof}}{\cite{review}, Appendix~A}.}
\section{Algorithms}
\label{sec.algorithms}
\iftoggle{showBudget}{\DS{This should be roughly 4 pages}}{}
In this section, we present a novel algorithm, \CWM, for the problem described in Section~\ref{sec.framework}.
\CWM takes as input (i) a DAG $G=(V,E,\omega, c)$ representing the workflow \new{and} a user-defined deadline, 
(ii) cluster specifications such as $\idleP{}$, $\workP{}$ and
the speed $s(\cdot)$ of each processor in the cluster, 
and (iii) a power profile with a green power budget per interval. 
From that, \CWM computes in two phases a carbon-optimized schedule that meets the deadline.
\new{First, it computes a deadline-agnostic preliminary schedule by focusing on making mapping decisions in a carbon-aware way, see Section~\ref{sec.deadline-agnostic-mapping}.
In the second phase, \CWM adjusts the schedule so that it meets the deadline (if needed),
while still taking carbon costs into account, see Section~\ref{sec.deadline-repair}.}
At the end of each phase,
we use a local search, described in Section~\ref{sec.local-search}, to further improve the solution.
\subsection{Deadline-Agnostic Mapping}
\label{sec.deadline-agnostic-mapping}
The overall idea for computing a carbon-optimized, deadline-agnostic schedule
in \CWM is the following. 
For each given interval, we have a power budget. Since the availability of green power varies over time, 
it is likely that there is not
enough green power for the whole cluster to be active at zero cost (at least not all the time). 
Hence, we solve for every interval a dynamic program to 
find a suitable subset of processors that can be active at the same time without exceeding the power budget. 
Within this subset of processors, we then use a version of the list-based scheduling heuristic \algo{HEFT}~\cite{HEFT} to optimize for makespan. 
The adaptation, which we call \algo{HEFT-SL}, is necessary to 
account for the carbon cost of communications.
Overall, our approach combines at every stage carbon awareness and deadline awareness. 
Yet, in this phase, \CWM
prioritizes carbon awareness for its decisions. 
First, it chooses for every interval $I_j$ a subset~$\mathcal{P}_j$ of active computation processors at zero carbon cost as follows.

\smallskip
\noindent{\bf Processor selection. }
To choose a processor subset $\mathcal{P}_j$ for an interval $I_j$ with green power budget $\budget_j$, 
we solve an optimization problem to select a subset of processors
that maximizes the combined processor speed while respecting the given 
green power budget.
\new{This is formulated as a variant of the $0-1$ knapsack problem where items are the available computation processors
whose weight is $\workP{}$, the value is the speed $s(\cdot)$ of the processor,
and the knapsack capacity is given by $\max\{0, \tau (\budget_j - \mathcal{P}_{\text{base}})\}$.
The base power $\mathcal{P}_{\text{base}}$ is given by $\mathcal{P}_{\text{base}} = \sum_{i=1}^{P^2} \idleP{i}$.}
\iftoggle{showOld}{\old{
\begin{itemize}
    \item items are the available computation processors;
    \item the weight of an item is the dynamic power consumption of the processor;
    \item values are the speed values of the processors;
    \item the knapsack capacity is the given green power budget minus the base power $\mathcal{P}_{\text{base}}$ of the
    cluster, multiplied with a tuning parameter $\tau \in (0,1]$, where 
    \begin{equation*}
        $\mathcal{P}_{\text{base}} = \sum_{i=1}^{P^2} \idleP{i}.$
    \end{equation*} 
    The overall knapsack capacity is given by 
    \begin{equation*}
        $\max\{0, \tau (\budget_j - \mathcal{P}_{\text{base}})\}$.
    \end{equation*} 
\end{itemize}}}{}
To solve the corresponding dynamic program, we use a dynamic programming (DP) table and 
reconstruct the solution by backtracking through the table.
If the capacity is~$0$ (base power larger than green power budget), 
we choose $\mathcal{P}_j = \{p_{\min}\}$, where $p_{\min}$ is one of the processors with the smallest dynamic power consumption:
$\workP{p_{\min}} = \min_{1\leq k \leq P} \workP{p_k}$. 
This is to ensure that potential bottleneck tasks in the workflow can still be scheduled 
even if there is not much green power available. Thereby, and by choosing the fastest processor subsets, 
we incorporate some deadline-awareness in this phase.
\iftoggle{showOld}{\old{One reason behind introducing the tuning parameter $\tau$ for the dynamic program is as follows.
Note that we can have the situation that one task starts in a given interval $I_j$ but ends in interval $I_{j+1}$.
In this scenario, if the processor subset of interval $I_{j+1}$ is active in parallel, there would potentially be more
processors active at the same time than intended. Hence, to mitigate this effect, we multiply the potential green power budget by $\tau \leq 1$ to make
sure that we have some buffer.
An example of such a situation is shown in Figure~\ref{fig.dp_tuning_parameter}.}}{}
\new{The parameter $\tau\leq 1$ buffers against tasks that cross interval boundaries, preventing budget overruns when there are more processors
active than planned for the subsequent interval the tasks run into.
An example is shown in~\iftoggle{TR}{Appendix~\ref{appendix.algorithm-details.additional_figures}, Figure~\ref{fig.dp_tuning_parameter}}{\cite{review}, Appendix~B.1}.}
Overall, we obtain a map $\mathcal{P_{\_}} : \{I_1,\dots, I_J\} \longrightarrow \{\mathcal{P}_1, \dots, \mathcal{P}_J\}$, 
where $\mathcal{P}_j$ is the chosen processor subset for interval $I_j$.
We also provide pseudocode for this phase in~\iftoggle{TR}{Appendix~\ref{appendix.algorithm-details.further_pseudocode}, Algorithm~\ref{alg:processor-subset}}{\cite{review}, Appendix~B.3, Algorithm~1}. 

\smallskip
\noindent{\bf Initial Mapping and Scheduling. } 
First, \CWM ranks the tasks in the workflow using a similar definition
as the bottom levels in HEFT~\cite{HEFT}, i.\,e., we compute for every task $v\in V$:
    $\text{rank}(v) = \overline{\text{rt}(v)} + \max_{(v,w)\in E}(c(v,w) + \text{rank}(w)),$
where $\overline{\text{rt}(v)}$ is the mean running time of task $v$ in the cluster.
\iftoggle{showOld}{\old{Note that we use the edge weight of the edge instead of the mean communication time
(as in HEFT) because we assume a uniform communication bandwidth $\beta=1$.
The order in which we compute the rank values is given by a reverse topological order of the tasks.}}{}
Afterwards, we sort the tasks according to their rank in descending order.
To break ties between vertices with the same upward rank, we shuffle the tasks randomly before sorting.
Then, the algorithm iterates over the vertices in the order given by the bottom levels. 
Let $v\in V$ be the current task.
First, we have to determine which (preliminary) interval to use.
For this, we look at the earliest start time of~$v$, ignoring communication:
%
    $\text{EST}'(v) = \max_{(u,v)\in E}\left(\finish{u}\right)$.
%
Afterwards, we find the interval $[b_j, e_j) = I_j\in\{I_1,\dots,I_J\}$, such that $ b_j \leq \text{EST}'(v) < e_j$.
(We ignore the communication in this step since we would first have to find suitable spots for the 
messages -- which would be time-consuming.)
Then, \CWM tries to find an assignment for task~$v$ using the processor subset~$\mathcal{P}_j$.
For this, we follow a \algo{HEFT}~\cite{HEFT}-like procedure and iterate over each processor $p\in \mathcal{P}_j$. 
For each of these processors, we first look at the input arrivals $A^v_p$ of each predecessor $u$ of $v$. 
If $\mu(u) = \mu(v)$, no communication is necessary and hence we add the finish time of $u$ to the arrivals, i.\,e., 
%
    $A^v_p = A^v_p \cup \{\finish{u}\}.$
%
\iftoggle{showOld}{
\old{If $v$ and $u$ are located on different processors, we have to communicate the output of $u$
to the processor $\mu(v)$ of $v$.
By definition of our model, we have to find a suitable spot on the communication link for the connection $\mu(u) \rightarrow \mu(v)$.
We denote this link by $p_{(\mu(u), \mu(v))}$. Algorithmically, we have the following challenge here.
Since for each link we can only schedule one communication at a time, we have to find a suitable spot on the link.
Thus, we cannot only look at the maximum finish time of each predecessor message, but we must tentatively schedule the communication 
on the link to find the earliest possible and valid arrival time for the message from $u$ to $v$.
Such a situation is shown in Figure~\ref{fig.communication_scheduling}.
}}
\new{If $u$ and $v$ are on different processors, we must find a valid slot on the link $p_{(\mu(u),\mu(v))}$; since links are serialized, we tentatively schedule each message to obtain its true earliest arrival time.
Note that we cannot just look at the maximum finish time for each predecessor message since this might cause overlaps.
An example for such a situation is shown in~\iftoggle{TR}{Appendix~\ref{appendix.algorithm-details.additional_figures}, Figure~\ref{fig.communication_scheduling}}{\cite{review}, Appendix~B.1}.}
\new{To solve this challenge, we introduce copies of the actual communication links that we only create on demand. 
Otherwise, the link copies only exist as references to avoid unnecessary copies of the $O(P^2)$ communication links.
On these link copies, we then actually schedule the communications to find valid spots for them.}
Once we found the spot for the communication \new{$c_{u,v}=(u,v)\in E$}, we add the finish time to the arrivals, i.e., \new{$A^v_p = A^v_p \cup \{\finish{c_{u,v}}\}$.}
The order in which we schedule the messages is just the order in which we store them in memory.
Once we have gathered all arrivals for the predecessors of $v$, we have the real earliest start time 
%
\new{$\text{EST}(v) = \max_{(u,v)\in E}\left(\finish{c_{u,v}}\right) = \max_{t \in A^v_p}(t)$,
%
where $\finish{c_{u,v}} = \finish{u}$ if $\mu(u) = \mu(v)$.}
To restore the best decision later on without recomputing the communication spots, we store all the data transfers for this processor candidate.

Now that we have the earliest start time for $v$, we can find a spot for the task on the corresponding processor. 
For this, we look at the earliest gap on the processor that fits the task and respects $\text{EST}(v)$. 
In~\cite{HEFT}, this is called \textit{insertion-based strategy};
it yields the \textit{earliest finish time} $\text{EFT}_p(v)$ of task~$v$ on processor~$p$. 
The processor $p_v$ on which $v$ is scheduled is then simply given by the processor that minimizes $EFT_p(v)$ over all
candidate processors in $\mathcal{P}_j$.
Ties for $EFT_p(v)$ are broken uniformly at random.
Now, we have a preliminary start time $\start{v}'$ and a preliminary processor assignment $p'$ for task $v$. 
Note, however, that we ignored the communications for the processor subset selection. 
That is why we check whether the preliminary start time is indeed in the corresponding interval~$I_j$, i.e., we check if 
$b_j \leq {}\start{v}' < e_j$. 
If not, we look at the next interval~$I_{j+1}$ and the corresponding processor subset $\mathcal{P}_{j+1}$. 
Then, we repeat the same procedure as for interval~$I_j$, with the additional restriction that the task is not allowed to start
earlier than~$b_{j+1}$. This procedure is repeated at most three times or until we reach the end of the 
simulated time horizon. Preliminary experiments with other numbers of repetitions yielded no improvements.
After that, we just accept the solution found to limit the necessary
running time for scheduling and its resulting makespan.
At the end of this routine, we have a preliminary schedule,
which is subsequently refined with a local search described in Section~\ref{sec.local-search}.
Pseudocode for this procedure can be found in~\iftoggle{TR}{Appendix~\ref{appendix.algorithm-details.further_pseudocode}, Algorithm~\ref{alg:initial-mapping}}{\cite{review}, Appendix~B.3, Algorithm~2}.
\subsection{Deadline Repair}
\label{sec.deadline-repair}
In this section, we describe what we do if the initial schedule does not meet the given deadline~$D$. 
The intuition behind this routine is the following. When restricting the usable processors in the previous step, we limited the parallelism in the workflow.
Even though we still aimed for a short makespan, this might have led to deadline violations.
Hence, we now try to increase parallelism in the schedule by lifting the constraints on the processor subsets.
\iftoggle{showOld}{\old{First, we compute bottom levels again to get a ranking of the tasks.
This is done as described in Section~\ref{sec.deadline-agnostic-mapping} to prepare
for the subsequent rescheduling of some of the tasks.}}{}
\new{First, we re-rank tasks as in Section~\ref{sec.deadline-agnostic-mapping}.} 
In order to compute the set of tasks that we want to reschedule,
we look at all tasks $v$ that finish after a given threshold $\xi$, i.e., $\finish{v} > \xi$,
and add them to a set $\mathcal{R}'$. Additionally, we iterate over each task $v\in \mathcal{R}'$ and
add all its successors to $\mathcal{R}'$ via depth-first search.
This is necessary because rescheduling a task might influence its successor tasks.
Hence, the set of tasks that we want to reschedule is given by
%
    $\mathcal{R} = \{u \in V \ | \ \exists \ v \in \mathcal{R}' \text{ such that } \exists \text{ path } v \rightarrow u\}$.
%
This set of tasks is rescheduled now without any restrictions on the set of processors, i.\,e., we use \texttt{HEFT-SL}.
Here, we use the order of the tasks given by the ranking that we computed at the beginning of the routine.
Furthermore, we respect the assignments of all tasks $V \setminus \mathcal{R}$ and combine the schedules.
In the first iteration of this routine, we set $\xi = D$. If this is not sufficient to reach the deadline, 
we do a binary search over the deadline from $0$ to $D$, to find the largest threshold~$\xi$ such that
the deadline is met. 
\iftoggle{showOld}{\old{Note that it is in general desirable to exploit the whole deadline since this gives us the most
flexibility for carbon-aware decisions.}}{}
Note that if $\xi = 0$, this is essentially $\heft$, which is our fallback solution.
Therefore, the deadline that is achievable by $\heft$ (with the given random seed for tie-breaking) is the tightest
deadline that $\CWM$ can achieve.
If we meet the deadline, we additionally apply the local search algorithm described in Section~\ref{sec.local-search}.
We provide pseudocode for this procedure in~\iftoggle{TR}{Appendix~\ref{appendix.algorithm-details.further_pseudocode}, Algorithms~\ref{alg:deadline-repair},\ref{alg:rescheduling}}{\cite{review}, Appendix~B.3, Algorithms~4,5}.
\subsection{Local Search}
\label{sec.local-search}
The local search routine further refines the solution, once after the initial scheduling (see Section~\ref{sec.deadline-agnostic-mapping}), and once more after the deadline repair (see Section~\ref{sec.deadline-repair}).
We identify intervals where we use carbon-emitting power 
and try to shift tasks from such intervals to other ones.
The local search routine runs in several rounds;
a parameter $\phi$ defines the maximum number of iterations. 
\new{In each iteration, we identify the leftmost interval $I$ where
we exceed the green power budget.
From this interval, we shift a randomly selected task $v'$ together with 
all its successors, later-scheduled tasks and corresponding communications
by the minimal offset required to place $v'$ beyond $I$, while respecting the given deadline $D$.
For a more detailed description of this routine, we refer to 
\iftoggle{TR}{Appendix~\ref{appendix.algorithm-details.local_search_details}}{\cite{review}, Appendix~B.2}};
for pseudocode, we refer to \iftoggle{TR}{Appendix~\ref{appendix.algorithm-details.further_pseudocode}, Algorithm~\ref{alg:local-search}}{\cite{review}, Appendix~B.3, Algorithm~3}.
\iftoggle{showOld}{\old{
At each iteration, we first compute a set of refined intervals that correspond
to beginning/end of tasks, as well as the carbon cost in each of these intervals
(see details in~\cite{CaWoSched}, Appendix~A.1). 
Then, we iterate over the intervals from left to right. If we find an interval where the total power consumption exceeds the budget, this becomes 
the interval that we want to resolve. We denote this interval by $I=[b, e)$.
If we do not find one, we end the routine since the carbon cost is zero in this case. 
As already mentioned in Section~\ref{sec.deadline-agnostic-mapping},
\iftoggle{showOld}{\old{and visualized in Figure~\ref{fig.dp_tuning_parameter}}}{}
one main reason for exceeding the green power budget for a given interval is that some tasks 
start before the interval but continue running during it, i.\,e., 
tasks $v$ with  
    $\start{v} < b$ and
    $\finish{v} > b$.  
Hence, we look at \emph{all} tasks that run during this interval: 
$\mathbf{C}_I = \{v\in V \ | \  \start{v} < e \ \text{and} \ \finish{v} > b \}.$
We choose one task $v'$ uniformly at random in~$\mathbf{C}_I$, which we aim at shifting 
to a later interval to reduce the power consumption in interval~$I$.
However, we have to ensure that the shift does not violate the 
validity of the schedule. We also consider the tasks {\em following}~$v'$, 
from the task set\new{, i.e., all tasks that start later than $e$ and all successors of $v'$.}
\iftoggle{showOld}{\old{
\begin{align*}
    \mathbf{L}'_I & = \{v\in V \ | \  \start{v} \geq e\}  \\
    &\cup \{ v\in V \ |\   v\in\text{succ}(v') 
    \ \text{and} \ \start{v} < e\},
\end{align*} 
where $\text{succ}(v')$ is the set of all successors of $v'$.}}{}

Furthermore, we have to ensure that we do not cause overlaps on the processors after shifting tasks.
Hence, we look at every task $v \in \mathbf{L}'_I$ and add all tasks $u$ such that $\mu(u) = \mu(v)$
and $\start{u} \geq \finish{v}$ to a new set~$\mathbf{L}_I$. 
\iftoggle{showOld}{\old{Before we proceed, we perform $\mathbf{L}_I := \mathbf{L}_I \cup \mathbf{L}'_I$.}}{}
\new{We then set $\mathbf{L}_I := \mathbf{L}_I \cup \mathbf{L}'_I$.}
Afterwards, we determine the time units by which we shift tasks. Hence, we look at the minimum move that 
shifts $v'$ outside of the interval $I$, i.\,e.,
%
    $s_{v'} = e - \start{v'}.$
%
We also have the possibility to incorporate the given deadline~$D$ into the local search routine.
This becomes important when we refine the solution after the deadline repair step. 
If this is the case, we first compute the latest ending task in the schedule and denote it by~$v_l$. 
Then, we adapt the move to be 
%
$s_{v'} = \min\{(e - \start{v'}), (D-(\finish{v_l}))\}$.
%
To be computationally efficient, the idea is now to shift the whole set of tasks $\{v'\} \cup \mathbf{L}_I$ by the same move $s_{v'}$.
Additionally, we have to shift the corresponding communications as well. Hence, we look at each 
outgoing message of each task in $\{v'\} \cup \mathbf{L}_I$ and shift the communication by $s_{v'}$, too. 
This way, we avoid an expensive rescheduling of each communication individually.
\new{We prove the correctness of this approach in \iftoggle{TR}{Appendix~\ref{appendix.proof-local-search-correctness}, Lemma~\ref{lemma.local-search-correctness}.}{Appendix~A.}}
\iftoggle{showOld}{\old{
is shown through the following Lemma~\ref{lemma.local-search-correctness}
(\iftoggle{TR}{proof in Appendix~\ref{appendix.proof-local-search-correctness}}
{proof in \cite[Appendix~A]{review}}):
\begin{lemma}\label{lemma.local-search-correctness}
Let $s\in \mathbb{N}$ and let $\sigma$ be a valid schedule (see Section~\ref{sec.framework}).
Furthermore, let $v' \in V$ and $\mathbf{L}_I$
be defined as above. Shifting all tasks in $\{v'\} \ \cup \ \mathbf{L}_I$ and the corresponding outgoing communications by $s_{v'}$
transforms $\sigma$ into another valid schedule $\sigma'$.
\end{lemma}}}{}
}}{}
\section{Experimental Evaluation}
\label{sec.experiments}
\iftoggle{showBudget}{\DS{This should be roughly 4.5 pages (incl. pictures)}}{}
We evaluate the proposed scheduling algorithm $\CWM$ with an extensive set of simulation experiments against multiple competitors. 
\iftoggle{showOld}{\old{As a first carbon-agnostic baseline, we consider \texttt{HEFT-SL}, which is a version of the very popular \texttt{HEFT} algorithm~\cite{HEFT} that we adapted
with minimal changes to make it work in our communication framework.}}{}
%
\new{To this end, we adapted \texttt{HEFT}~\cite{HEFT} to our communication model; we denote the resulting 
carbon-agnostic algorithm by \heft.}
For the comparison to the state of the art, we include two configurations of our previous work \algo{CaWoSched}~\cite{CaWoSched};
it is the only competitor we know of that also optimizes for carbon cost within a comparable framework,
even though it further considers that the mapping and ordering is fixed. 
We hence use the initial mapping produced by \algo{HEFT-SL}, 
and the \algo{CaWoSched} code published at 
\url{https://github.com/KIT-EAE/CaWoSched}.
\new{We use the same parameter configuration as the original work~\cite{CaWoSched}, details can be found there. 
Furthermore, we use the variants \texttt{H-CWS-p} and \texttt{H-CWS-s}, where $-\texttt{p}$ and 
$-\texttt{s}$ represent the base score \textit{pressure} and \textit{slack}, respectively.}
Note that the abbreviations of the algorithm's names are introduced here for clarity.

We implemented all algorithms using C++ and compiled them with g++ (v.13.2.0) -O3.
The code used for the simulations is publicly available for reproducibility purposes 
\iftoggle{TR}{at \url{https://github.com/KIT-EAE/CWM.git}}{in the Zenodo repository~\cite{review}}. 
We provide all compiler options there. 
The simulations are executed on workstations with CentOS~8, 192 GB RAM, 
and $2\times 12$-Core Intel Xeon 6126 at 3.2 GHz.
To manage the experiments, we use the tool simexpal~\cite{simexpal}.
\subsection{Experimental Setup}
\label{sec.setup}
%
\iftoggle{showOld}{\old{The experimental setup is similar to the setup in~\cite{CaWoSched} since it corresponds to our main
competitor \algo{CaWoSched}:}}{} 
\new{We follow the setup of~\cite{CaWoSched}:}
\iftoggle{showOld}{\old{
\smallskip
\noindent{\bf Target computing platform. }
In general, all algorithms can handle a heterogeneous computing platform (in terms of memory size and processor speed).
Hence, we simulate two different heterogeneous computing platforms.
For ease of notation, we refer to them as the \textit{small cluster} and
the \textit{large cluster}. Each cluster contains multiple heterogeneous compute nodes,
where each compute node has an idle and an active power and a speed value as described
in Section~\ref{sec.framework}.
Their values are inspired by real-world machines\footnote{We use the following values from \url{https://www.spec.org/power_ssj2008/results/res2025q1/}: 'Average Active Power (W)' at a target load of $100\%$ 
for~$\workP{}$, 'Average Active Power (W)' at active idle for $\idleP{}$, and the 'ssj\_ops' 
(Server Side Java Operations) value divided by $10000$ (for normalization) as the speed value $s$.}. 
We use six different specification types; for the small [large] cluster, we use $12$ [$24$] nodes of each type, i.\,e., 
a total of $72$ [$144$] nodes. 
Since the communication links are usually responsible for a much smaller fraction of the total power 
usage~\cite{comm_links_justification}, we generate their idle and work powers with 
much smaller power values by a normal distribution. The bandwidth is set to $\beta=1$.

\smallskip
\noindent{\bf Workflows. }
%
In line with Section~\ref{sec.framework}, we represent the workflows as DAGs. 
To create them, we use the same tools and techniques as in~\cite{CaWoSched}:
we first take real-world workflows (atacseq, bacass, methylseq -- only a rescaled version is used, eager, chipseq) 
from~\cite{lotaru}, transform them to our model, and scale them with the WFGen generator~\cite{wfcommons} 
to the desired size.  The resulting 44 workflows have between $12$ and $\numprint{30000}$ vertices.  
The DAG weights are sampled 
from a normal distribution, where mean and standard deviation are derived from the speed values of the cluster. 

\smallskip
\noindent{\bf Profiles. }
%
For the power budget profiles, we use a different approach than~\cite{CaWoSched}.
We sample real carbon intensities from an Electricity Maps dataset~\cite{electricitymaps}, using the hourly data for Germany and for California in 2024. 
A given time horizon is first partitioned into intervals by computing random split points.
Thereby, we ensure that each interval length is uniform within a configurable integer range (we chose $[10, 50]$). 
For the intervals, we pick a contiguous subsequence of carbon intensity values uniformly at random. 
Depending on the cluster type, we then rescale the values as follows to invert the carbon
intensities, making the smallest carbon intensity correspond to the largest green power budget:
For the sequence of values, we denote by $[x_{\min}, x_{\max}]$ the range of carbon intensities. 
Then we map this range to a target power interval $[P_{\min}, P_{\max}]$, where we choose 
$P_{\min}$ to be the sum of all idle powers in the cluster and for $P_{\max}$ we add 
a fraction ($0.2$ for the small cluster, $0.4$ for the large cluster) of the sum of all dynamic power values
for the compute nodes and the communication links.
Then, any carbon intensity $x$ becomes 
\begin{equation}
\label{eq:carbon-intensity-power-budget}
    x^{\prime} = P_{\max} - \frac{x-x_{\min}}{x_{\max} - x_{\min}}(P_{\max} - P_{\min}).
\end{equation}
%

\smallskip
\noindent{\bf Further setup. }
%
$\algo{CWM}$ uses $\phi = 500$
for the local search proposed in Section~\ref{sec.local-search};
also, we set $\tau = 0.8$ [$\tau=0.6$] for the small [large] cluster (see Sec\-tion~\ref{sec.deadline-agnostic-mapping}).
As tightest deadline, we consider the makespan of the schedule produced by \algo{HEFT-SL}, denoted by $M$. Then for each 
algorithm $A \in \{\algo{CWM}, \algo{H-CWS-p}, \algo{H-CWS-s}\}$, we use as deadlines $D = \alpha \times M$ with
different scaling factors $\alpha > 1$ for more flexibility:
\begin{equation}\label{eq.deadlines}
    D \in \{1.2 \times M, 1.5 \times M, 2.0\times M\}.
\end{equation}

We do not consider the tightest setting where $\alpha=1$, since
$\algo{CaWoSched}$ uses a specific 
way of ordering communications and may not obtain
a schedule respecting the deadline \mbox{$D = M$}. 
Furthermore, 
there is an inherent trade-off between saving carbon emissions (our main goal) and
tighter deadlines.

With this setup (of clusters, workflows, profiles, deadlines), we test each algorithm on 
$2 \times 44 \times 2 \times 3 = 528$ instances. 
}}{}
\begin{table}[t]
\centering
\caption{Experimental setup summary. $M$: makespan computed by \heft.}
\label{table.setup}
\begin{tabular}{l || l r}
\hline
\textbf{Parameter / Input} & \textbf{Small cluster} & \textbf{Large cluster} \\
\hline
Nodes (6 spec types)     & $6\times 12 = 72$  & $6 \times 24 = 144$ \\
\hline
Number of Workflows      & \multicolumn{2}{c}{44} \\
\hline
Workflow sizes           & \multicolumn{2}{c}{12--\numprint{30000}} \\
\hline
Workflows                & \multicolumn{2}{c}{atacseq, bacass, methylseq, eager, chipseq~\cite{lotaru,wfcommons}} \\
\hline
Profiles                 & \multicolumn{2}{c}{Germany 2024, California 2024~\cite{electricitymaps}, hourly} \\
\hline
Interval lengths         & \multicolumn{2}{c}{$[10, 50]$ time units} \\
\hline
Deadlines $D$            & \multicolumn{2}{c}{$\alpha \times M$ for $\alpha \in  \{1.2,\, 1.5,\, 2.0\}$} \\
\hline
$\tau$ (processor selection) & $0.8$ & $0.6$ \\
\hline
$\phi$ (local search)    & \multicolumn{2}{c}{$500$} \\
\hline
Bandwidth $\beta$        & \multicolumn{2}{c}{$1$} \\
\hline
\hline
Total instances          & \multicolumn{2}{c}{$2 \times 44 \times 2 \times 3 = 528$} \\
\hline
\end{tabular}
\end{table}
\new{The key parameters are summarized in Table~\ref{table.setup}. Node power and speed values are derived from 
real-world SPEC benchmarks\footnote{From \url{https://www.spec.org/power_ssj2008/results/res2025q1/} as of December 2025:
\textit{Average Active Power} at 100\% load for $\workP{}$, at idle for $\idleP{}$,
and \textit{ssj\_ops}$/10000$ for $s$.}.
Communication link powers are drawn from a normal distribution with much smaller power values~\cite{comm_links_justification}.
DAG weights are sampled from a normal distribution, where mean and standard deviation are derived from the speed values of the cluster.
For more details on how we obtained these DAGs, we refer to~\iftoggle{TR}{Appendix~\ref{appendix.setup-details}}{\cite{review}, Appendix~C}.
The carbon intensities are inverted and linearly rescaled to a cluster-specific power interval
$[P_{\min}, P_{\max}]$ 
(see~\iftoggle{TR}{Appendix~\ref{appendix.setup-details}}{\cite{review}, Appendix~C}).
Furthermore, we exclude the case $\alpha=1$, since
$\algo{CaWoSched}$ uses a specific 
way of ordering communications and may not obtain
a schedule respecting the deadline $D = M$, and tighter deadlines inherently limit carbon savings.}
\subsection{Results}
%
We mainly focus  on the carbon cost of the  schedules produced by each algorithm, since all schedules
respect the target deadline. 
Also, we analyze the time \new{used by the algorithms} 
to compute the schedules and we discuss the trade-off 
between carbon cost and time. 

\smallskip
\noindent{\bf Algorithm Ranking. }
\label{par.ranking}
To rank  the different algorithms in terms of carbon cost,
we use performance profiles, following the Dolan-Mor\'e methodology~\cite{performance_profiles}. 
\iftoggle{showOld}{\old{
Given an instance $i$ and an algorithm $A\in \{\algo{CWM}, \algo{H-CWS-p}, \algo{H-CWS-s}, \heft\}$,
we consider its carbon cost $\cost_{i, A}+1$. Then, we look at the carbon cost ratio $r_{i,A} := (\cost_{i, A}+1) / (\cost^*_i+1)$,
where $\cost^*_i$ is the best carbon cost found for  instance~$i$.
Note that we shift the cost by one to avoid a division by zero if some algorithm finds an optimal solution.
While shifting might cause very large ratios if one algorithm finds the optimal solution and others do not,
this approach has in our experience no unfair influence for the competitors on the algorithm ranking.
After calculating the carbon cost ratios, we compute for each algorithm $A$ the fraction of instances
on which the carbon cost ratio of $A$ is within a factor $\delta$ of the best quality, i.\,e.,
\begin{equation*}
    \eta_A(\delta) = \frac{1}{N}|\{\text{Instances } i \ | \ r_{i, A} \leq \delta \}|.
\end{equation*}
These values are then plotted for a given set of thresholds~$\delta$.
}}{}
\new{Given an instance $i$ and an algorithm $A\in \{\algo{CWM}, \algo{H-CWS-p}, \algo{H-CWS-s}, \heft\}$,
we compute its carbon cost $\cost_{i, A}$ and we define $r_{i,A} := (\cost_{i,A}+1)/(\cost^*_i+1)$, where $\cost^*_i$ is the best cost found for instance~$i$. 
Note that costs are shifted by $+1$ to avoid division by zero. We report both geometric means and medians, 
as the latter are more robust to this shift.
Additionally, note that the absolute carbon costs have high values, so that adding $1$ is often negligible
if both algorithms compute schedules with positive carbon costs.}
%
\iftoggle{showOld}{\old{Again,
we shift the carbon cost by $1$ to avoid division by zero and to be able to compute geometric means.
The choice of the constant we add to the carbon cost influences
the geometric mean results. Thus, we do not only report the geometric means but also the medians; the latter are much
more robust when it comes to shifting results.  Additionally, note
that the absolute carbon costs have high values, so that adding $1$ is often negligible
if both algorithms compute schedules with positive carbon cost.}}{} 
%
%
%
\new{We then plot performance profiles $\rho_A(\tau) = \frac{1}{N}|\{i \mid r_{i,A}\leq\tau\}|$ for varying~$\tau$.}

In this section, we group the instances by the deadline~$D$.
A more detailed parameter study is provided later in this section. 
First, we look at the performance profile for the deadline $D=1.5\times M$ (left) in Figure~\ref{fig.performance-profile_abs-time}.

\begin{figure}[t]
\begin{subfigure}{0.45\textwidth}
\centerline{\includegraphics[width=\linewidth]{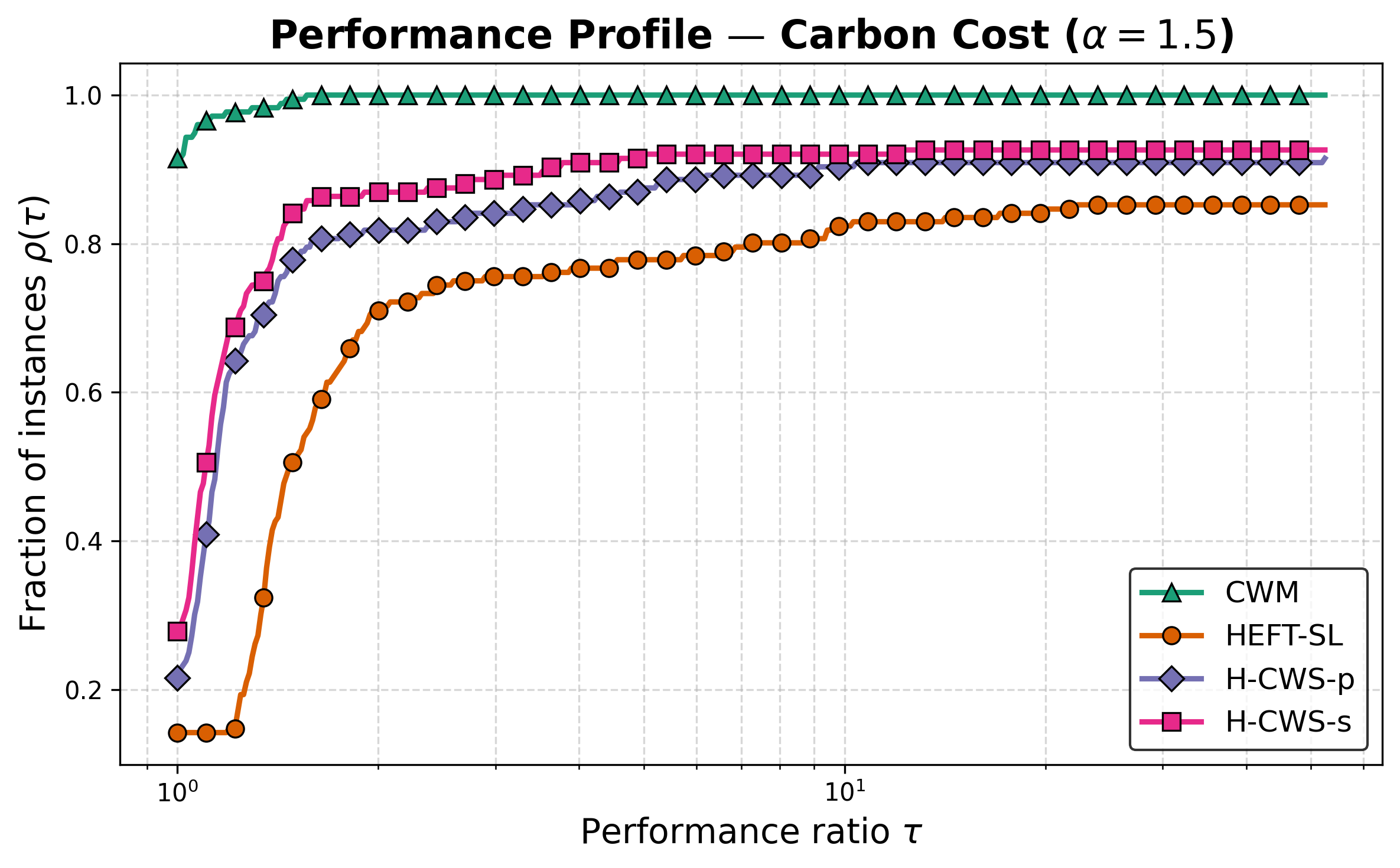}} 
\end{subfigure}
\hfill
\begin{subfigure}{0.46\textwidth}
\centerline{\raisebox{.3cm}{\includegraphics[width=\linewidth]{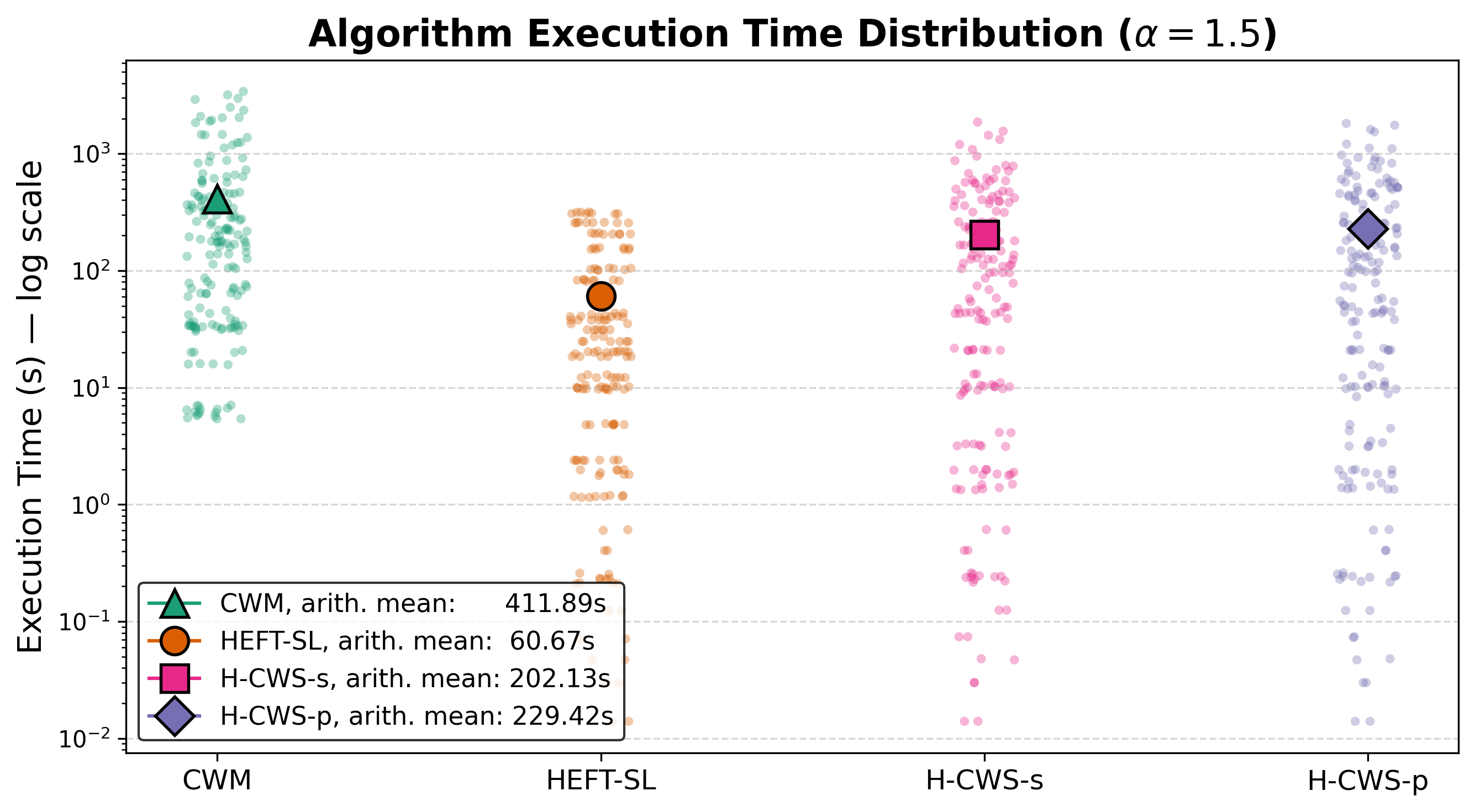}}}
\end{subfigure}
\caption{Performance profile (left) and absolute execution time in seconds (right) for deadline $D=1.5\times M$.}
\label{fig.performance-profile_abs-time}
\end{figure}
Since the \CWM curve is consistently at the top, our algorithm performs best. 
In particular, note that the best found solution is computed 
by $\CWM$ in $\approx 91.48\%$ of the instances.
Furthermore, in accordance with the experimental results in~\cite{CaWoSched}, we
see that $\texttt{H-CWS-s}$ and $\texttt{H-CWS-p}$ perform similarly -- with $\texttt{H-CWS-s}$
slightly better. As expected, the carbon-agnostic baseline $\texttt{HEFT-SL}$
yields the worst quality.
For a tighter deadline, \CWM, $\slack$, and $\pressure$ are closer to each other, while with
a larger flexibility in the deadline, the dominance of \CWM becomes even clearer, see~\iftoggle{TR}{Appendix~\ref{appendix.further-results}}{\cite{review}, Appendix~D}. 
This is to be expected, since there is less / more flexibility for the algorithms to optimize. 
%
%

\smallskip
\noindent{\bf Carbon cost analysis. }
\label{par.cost-analysis}
To quantify the findings from the algorithm ranking,
we look at cost ratios when we compare 
$\CWM$ against the competitors. For an algorithm $A\in \{\slack, \pressure, \heft\}$ 
and an instance~$i$, we look at the carbon cost ratio $(\cost_{i, \CWM}+1) / (\cost_{i, A}+1)$. 
%
%
%
\begin{figure}[tb]
\centering 
\begin{subfigure}{0.45\textwidth}
  \centerline{\includegraphics[width=\linewidth]{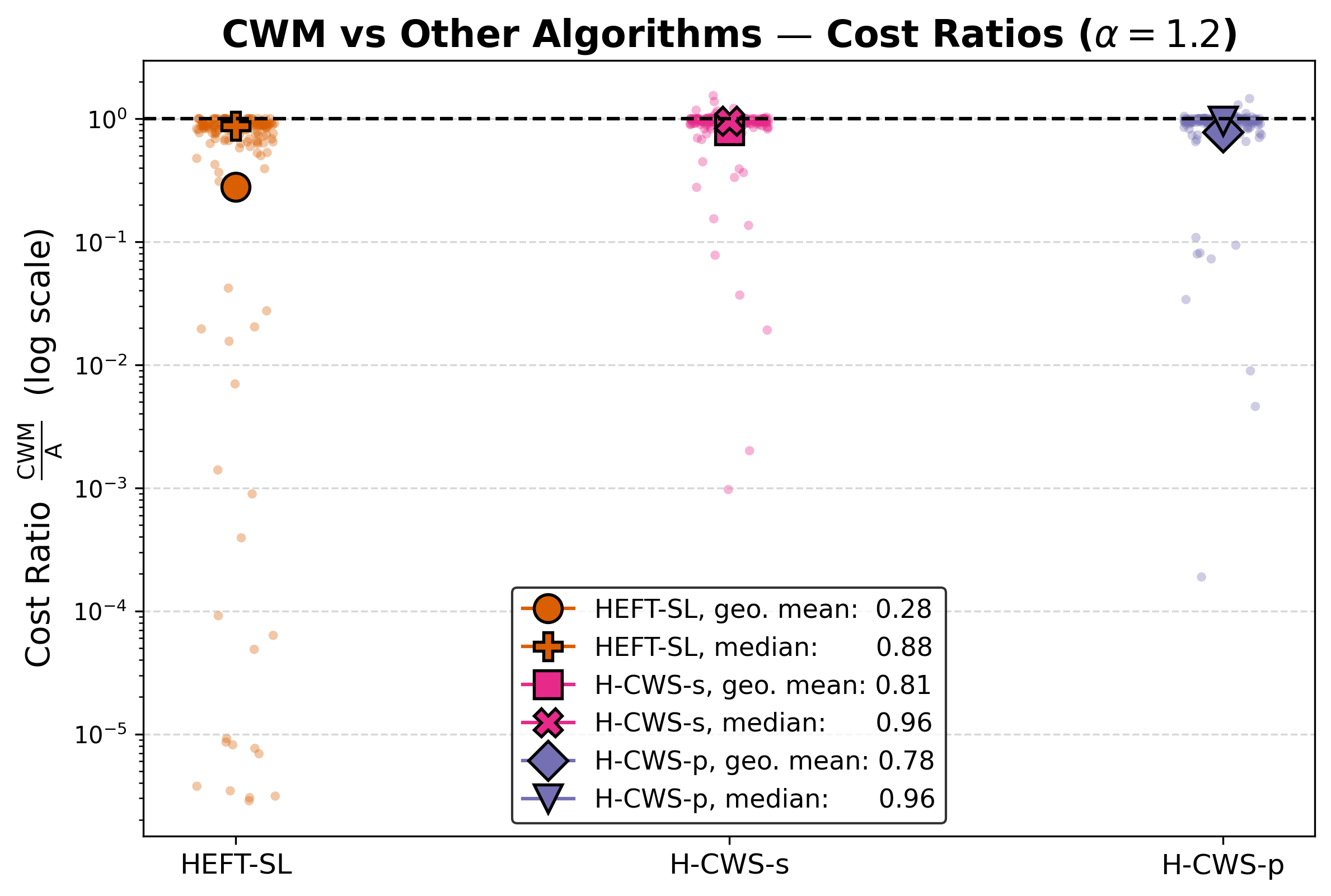}}
\end{subfigure}
\hfill
\begin{subfigure}{0.45\textwidth}
  \centerline{\includegraphics[width=\linewidth]{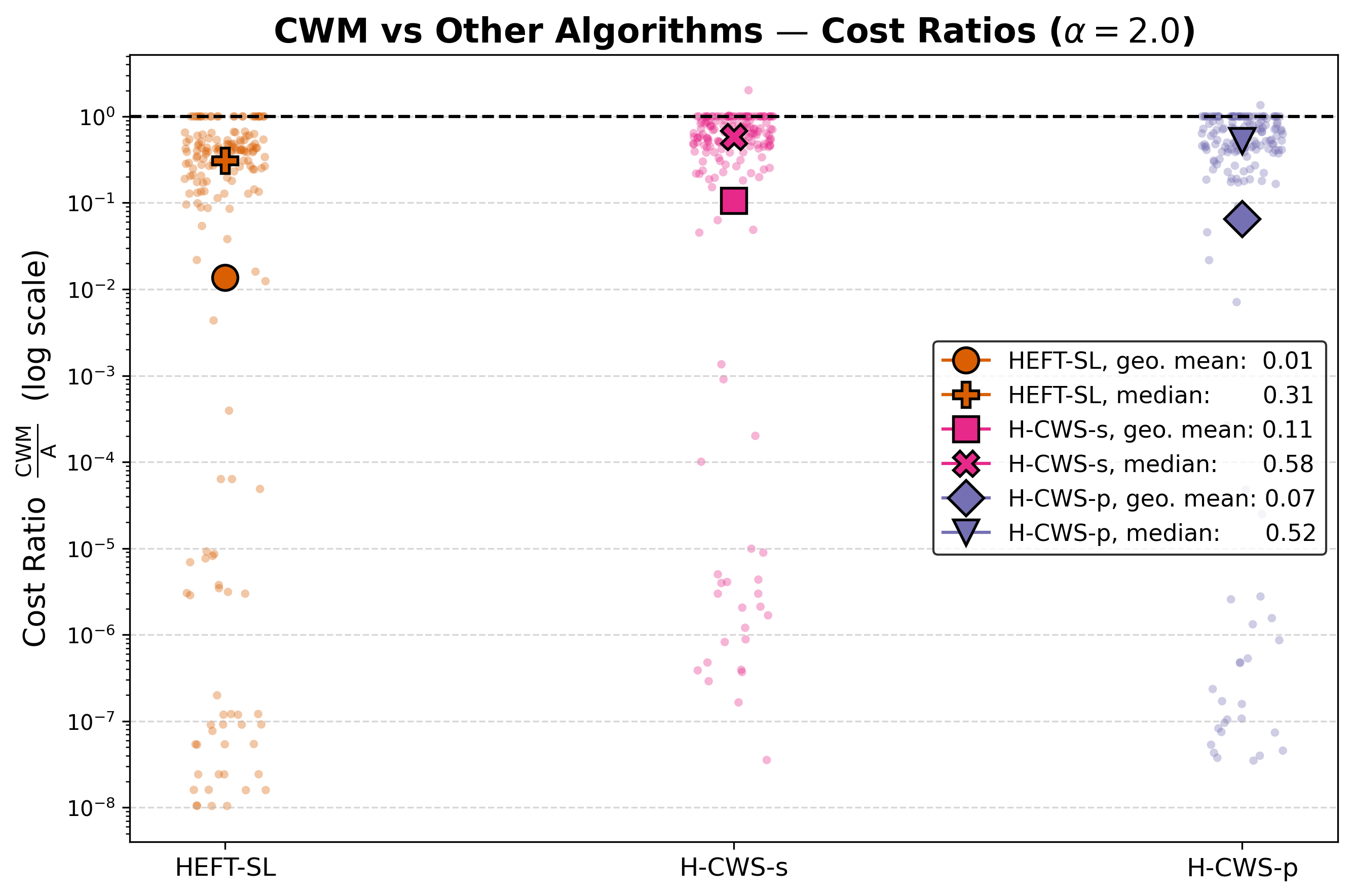}}
\end{subfigure}
\caption{Carbon cost ratios compared to $\CWM$ for deadlines 
$D=1.2\times M$ (left) and $D=2.0\times M$ (right).}
\label{fig.cost_ratios_12_20}
\end{figure}
%
We show the cost ratios for the deadline $D=1.2\times M$ in Figure~\ref{fig.cost_ratios_12_20} (left).
There are several instances for which $\CWM$ significantly improves over the other algorithms. 
This can also be seen by the geometric means (numbers are stated in the legend), 
which weighs the significant improvements more heavily. 
At the same time, the median improvement by \CWM over $\slack$ and $\pressure$ is rather small
(while improvements over $\heft$ are already noteworthy).
%
%
%
%
With higher flexibility, 
we see significant improvements also for the 
median, see Figure~\ref{fig.cost_ratios_12_20} (right). 
Compared to $\pressure$ and \slack, we almost halve the median carbon cost \new{($48\%$ and $42\%$ decrease, respectively).} 
\new{Note} that there are numerous instances for which $\CWM$ computes an optimal 
solution, while the others yield a non-zero carbon cost. 
In these cases, the geometric mean is less meaningful than the median. 
%
%
%
%
%
%
%
%
%
%
%

\smallskip
\noindent{\bf Time analysis. }
\label{par.time-analysis}
To briefly analyze 
the time the algorithms take 
to compute the schedules, 
Figure~\ref{fig.performance-profile_abs-time} (right) displays the total running time in seconds for 
the deadline $D=1.5\times M$. 
For a fair comparison, the running times of $\pressure$ and $\slack$ include
the time for computing an initial mapping with $\heft$.
\iftoggle{showOld}{\old{
\begin{figure}[tb]
\centering 
\begin{subfigure}{0.48\textwidth}
  \centerline{\includegraphics[width=\linewidth]{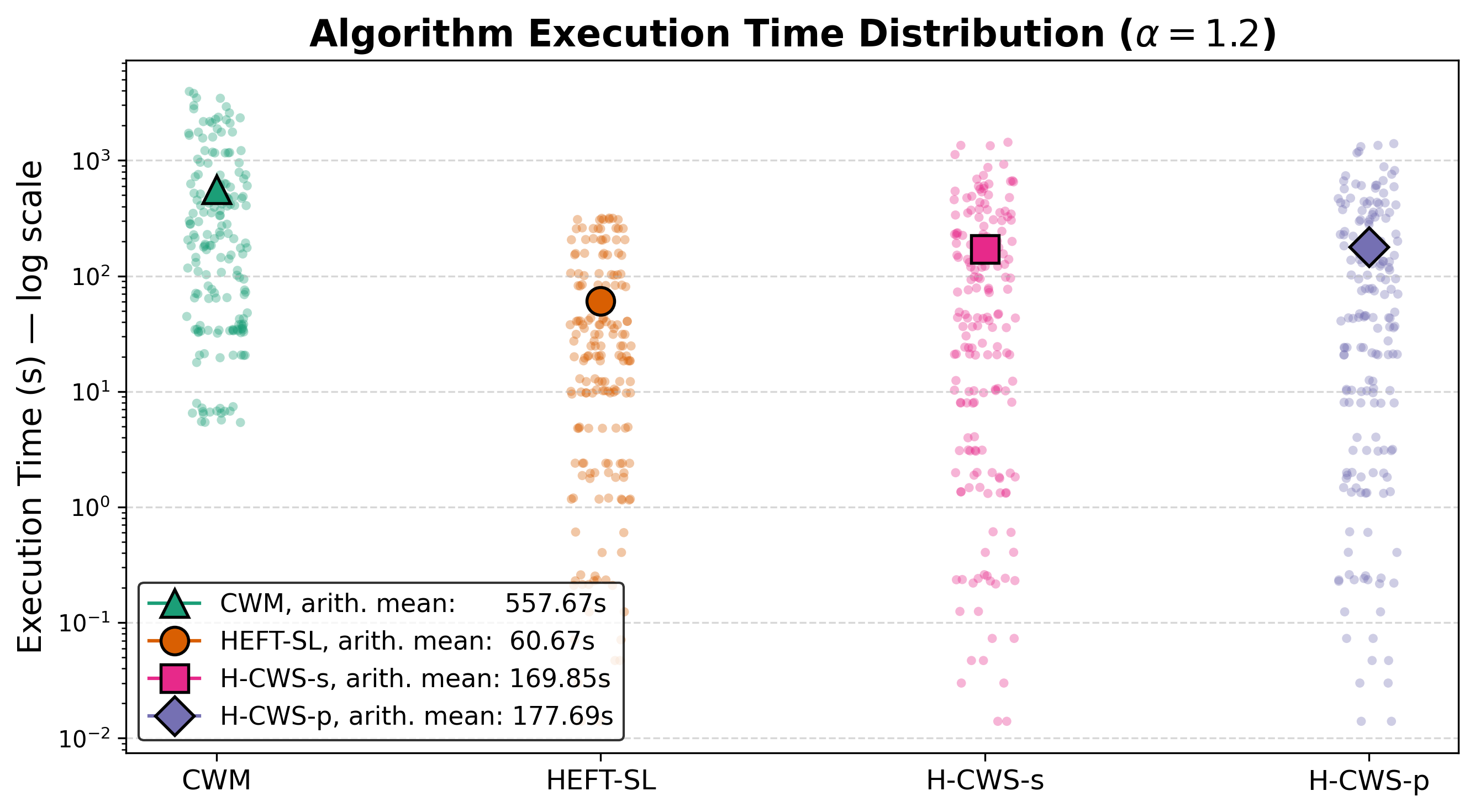}}
\end{subfigure}
\hfill
\begin{subfigure}{0.48\textwidth}
  \centerline{\includegraphics[width=\linewidth]{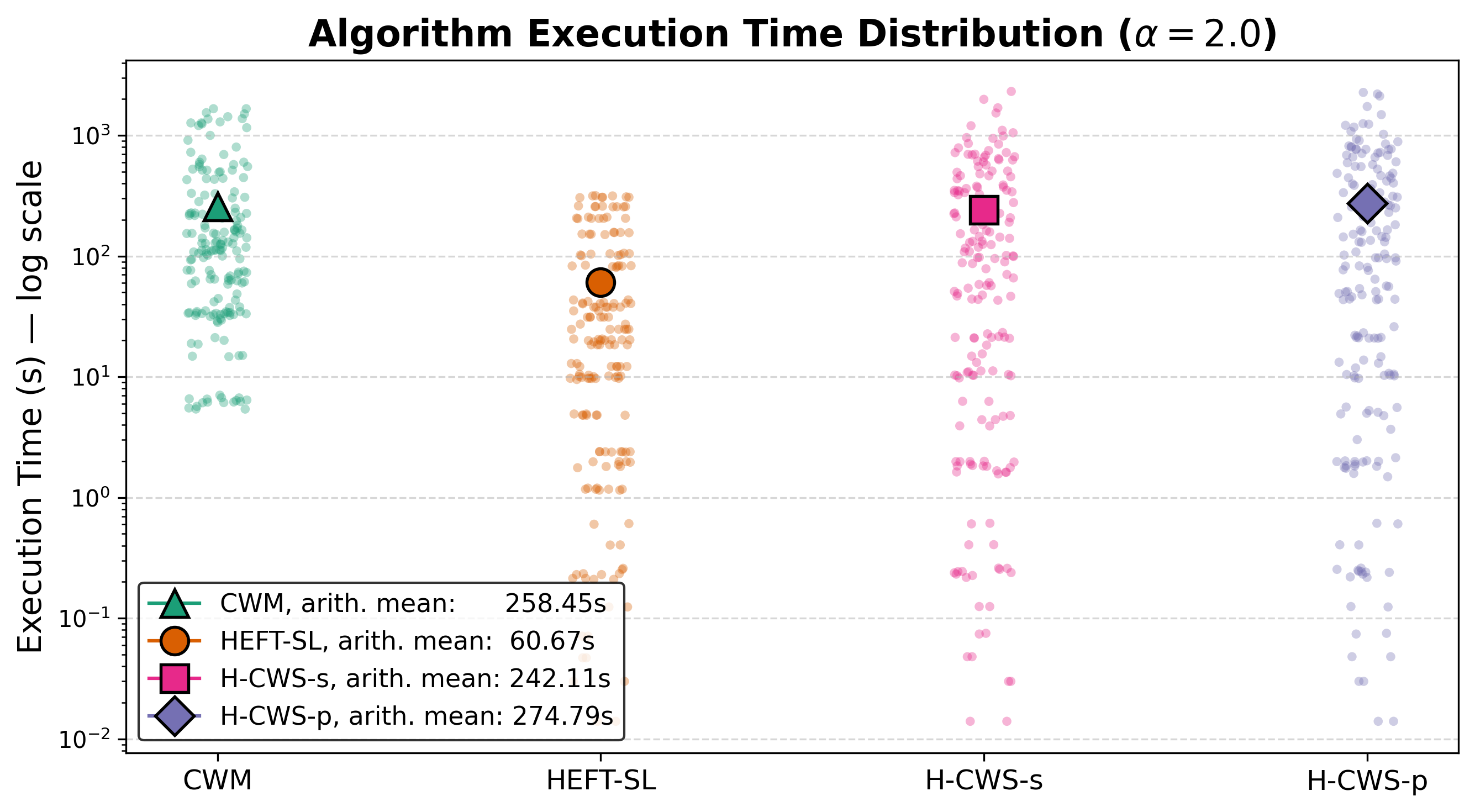}}
\end{subfigure}
\caption{Absolute execution time (seconds) for each algorithm for deadlines 
$D=1.2\times M$ (top) and $D=2.0\times M$ (bottom).}
\label{fig.abs_time_12_20}
\end{figure}
}}{}

For this deadline, $\CWM$ is roughly $2\times$ slower than the \algo{CaWoSched} algorithms
and more than $6\times$ slower than \heft.
This is to be expected since $\heft$ is carbon-agnostic and 
$\slack$ and $\pressure$ only shift tasks within a given mapping.
Moreover, note that $\CWM$ has to apply the repair routine of Section~\ref{sec.deadline-repair} 
more often for a tighter deadline than for larger deadlines. 
Hence, for $D=2.0\times M$, since the costly deadline repair does not have to be used, 
$\CWM$ is about as fast as $\slack$ and $\pressure$ (see\ \iftoggle{TR}{Appendix~\ref{appendix.further-results}}{\cite{review}, Appendix~D}), 
while $\CWM$ achieves much better carbon cost values in this case.
%

\smallskip
\noindent{\bf Parameter study. }
\label{par.parameter-study}
\iftoggle{showOld}{\old{
In this section, we present additional aspects of the experimental evaluations.
Due to space constraints, we only describe the results, 
while all corresponding plots
can be found in~\iftoggle{TR}{Appendix~\ref{appendix.further-results}}{\cite{review}, Appendix~D}.
In particular, we provide a more detailed analysis on the influence of the different
parameters such as the cluster size, the power profile, and the workflow size.
First, we show the additional plots regarding the deadline $D=1.5\times M$ for the performance profile, the cost ratios,
and the absolute execution times in \iftoggle{TR}{Appendix~\ref{appendix.further-results}, Figures~\ref{fig.performance-profile15},~\ref{fig.cost_ratios_15},~\ref{fig.abs_time_15}}{\cite{review}, Appendix~D}. 
There, one can observe that the trend for the deadlines $D=1.2\times M$ and $D=2.0\times M$ is confirmed:
$\CWM$ improves the competitors more with more deadline flexibility and the difference in running time
of $\CWM$ compared to $\slack$ and $\pressure$ becomes less if more deadline flexibility is given.
Also,
we show, for each deadline~$D$, plots for the absolute carbon cost,
including distributions in \iftoggle{TR}{Appendix~\ref{appendix.further-results}, Figures~\ref{fig.abs_cost_12},~\ref{fig.abs_cost_15} and~\ref{fig.abs_cost_20}}{Appendix~D}.
The same trend as for the cost ratios can be observed:
the more [less] flexible the deadline, the larger [smaller] the average carbon cost improvement yielded by \CWM.
Next, we can observe that the number of available compute nodes obviously influences the cluster utilization if the workflow is fixed.
That is why the cluster size indeed influences the relative performance of the algorithms.
For the deadline $D=2.0\times M$, the general trend is that \CWM achieves a higher relative improvement on the large cluster than on the small one\iftoggle{TR}{, see Appendix~\ref{appendix.further-results}, Figures~\ref{fig.cost_ratios_20_smallCluster} and~\ref{fig.cost_ratios_20_largeCluster}}{}. 
The effect is reversed only when comparing to \slack.
$\CWM$ is apparently more efficient when it comes to utilizing the additional processors.
This is to be expected, since $\CWM$ does not only focus on time shifts, but also it
considers the cluster size.
Moreover, on the large cluster, we observe a larger number of instances with an extremely low cost ratio for this deadline.
We conclude from this that for a larger cluster and $D=2.0\times M$, $\CWM$ achieves more optimal solutions than the competitors.
When the flexibility is lower ($D=1.5\times M$\iftoggle{TR}{, Appendix~\ref{appendix.further-results}, Figures~\ref{fig.cost_ratios_15_smallCluster} and~\ref{fig.cost_ratios_15_largeCluster}}{}), 
the number of extremely low cost ratios is also higher on the large cluster.
Yet, the median carbon cost ratios do not differ that much between the clusters.
Apparently, the lower flexibility limits the optimization capabilities of \CWM here.
The two different power profiles do not have a large influence
on the performance of the algorithms, for carbon cost ratios with $D=1.5\times M$\iftoggle{TR}{(see Appendix~\ref{appendix.further-results}, Figures~\ref{fig.cost_ratios_15_germany} and~\ref{fig.cost_ratios_15_california})}{}.
The only difference is that $\CWM$ seems
to find more optimal solutions for the Germany profile than for California in comparison to the competitors -- due to the data points
with extremely low cost ratio and the very low geometric mean.
Finally, we show the influence of the workflow size\iftoggle{TR}{ in Appendix~\ref{appendix.further-results}, Figures~\ref{fig.cost_ratios_15_realWorld}-%
\ref{fig.cost_ratios_15_large}}{} 
for the deadline $D=1.5\times M$. There, we can see the general trend
that $\CWM$ performs the better the larger the workflow is -- in comparison to its
competitors. For the tiny real-world workflows, we see almost no improvement.
Since they have at most $60$ tasks, most algorithms
find a carbon-optimal schedule. However, this is not the case
for the small ($\approx 200 - 4000$ tasks) and medium-sized ($\approx 8000 - 15000$ tasks)
workflows. While the respective average improvement of \CWM is rather modest (in terms of median), it becomes clearly better for large workflows ($\approx 20000 - 30000$ tasks).
}}{}
\new{Additional plots for all parameter influences are in~\iftoggle{TR}{Appendix~\ref{appendix.further-results}}{\cite{review}, Appendix~D}. 
In summary: (i)~the quality gains of \CWM over its competitors is confirmed by results for additional deadlines,
i.e.,~\CWM improves the competitors with more deadline flexibility\iftoggle{TR}{, see~Appendix~\ref{appendix.further-results}, Figures~\ref{fig.performance-profile_12_20},\ref{fig.cost_ratios_15},\ref{fig.abs-cost_12_15_20}}{};
(ii)~the difference in running time becomes smaller with more flexibility\iftoggle{TR}{, see~Appendix~\ref{appendix.further-results}, Figure~\ref{fig.abs-time_12_20}}{};
(iii)~\CWM benefits more from a larger cluster, finding more often solutions of carbon cost zero, 
especially at $D=2.0\times M$\iftoggle{TR}{, see Appendix~\ref{appendix.further-results}, Figure~\ref{fig.cost-ratios_cluster-size}}{};
(iv)~the two power profiles (Germany, California) yield similar relative rankings, with \CWM finding more 
solutions of carbon cost zero on the Germany profile\iftoggle{TR}{, see Appendix~\ref{appendix.further-results}, Figure~\ref{fig.cost-ratios_cal_ger}}{};
and (v)~\CWM's advantage grows with workflow size, from negligible gains on tiny workflows ($\leq 60$ tasks) to clearly superior performance on large ones ($\approx 20000$--$30000$ tasks)\iftoggle{TR}{, see Appendix~\ref{appendix.further-results}, Figure~\ref{fig.cost-ratios_size}}{}.}
%
\section{Conclusion}
\label{sec.conclusion}
\iftoggle{TR}{}{}
\iftoggle{showBudget}{\DS{This should be roughly 0.3 pages}}{}
We have presented a scheduling approach that takes carbon-aware decisions in both
the mapping and scheduling phases, while respecting a user-defined deadline. The scheduler 
restricts the parallelism in the cluster depending on the availability of green power, while it 
can also lift these constraints to find schedules that respect the deadline.
We also provide a local search approach that is capable of further improving the solution.
Experiments show that the proposed $\CWM$ algorithm  clearly dominates state-of-the-art solutions when
there is sufficient deadline flexibility. 
For example, for the largest level of flexibility, we improve $\texttt{CaWoSched}$ by decreasing their carbon cost by $42\%$ and $48\%$, respectively, 
depending on the configuration of their setup. This means that we nearly halve the carbon cost in this setting.
Compared to the carbon-agnostic baseline, we are able to decrease the carbon cost by~$69\%$. 
\iftoggle{showOld}{\old{
In future work, we plan to explore how to adapt the algorithm to an online setting
and how to handle uncertainty in task running times and the green power budget. 
It would also be interesting to consider shutting down some processors when 
there is enough slack in the schedule and a limited power budget; an updated model
and novel algorithms would be needed in this setting.}}{}
\new{Future work includes adapting the algorithm to an online setting, handling uncertainty 
in task running times and green power budget, and exploring processor shutdown strategies
when suitable schedule slack and power constraints are given -- the latter requiring an 
updated model and novel algorithms.}

\smallskip
\noindent{ \bf Acknowledgments. }
\begin{small}
This work is partially supported by Collaborative Research Center (CRC) 1404 FONDA --
\textit{Foundations of Workflows for Large-Scale Scientific Data Analysis}, 
which is funded by German Research Foundation (DFG).
%
The authors designed and implemented the algorithms and wrote the text; they used AI-based
tools (GitHub Copilot, ChatGPT, Claude) to support this process.
\end{small}

\bibliographystyle{splncs04}
\bibliography{references}


\iftoggle{TR}{
\appendix
\section{Proof of Theorem~\ref{thm.approx_DAG}}
\label{appendix.thmProof}
\begin{proof}
The (decision version of the) problem obviously is in NP: given a schedule (all tasks are mapped
onto the same processor) and a bound on carbon cost, it is easy to check in polynomial time 
that all constraints are respected, and that the carbon cost does not exceed the bound. 
The NP-hardness is obtained with a reduction from 3-partition~\cite{garey1979computers}: 
given an instance of 3-partition with $3n$ positive integers $a_{i}$ of total sum $nB$, we build an instance with 
$3n$ tasks of duration $a_{i}$ and a single unit-speed processor with a power profile consisting of
$n$ intervals of size~$B$ having a green power budget of~$1$, separated by $n-1$ intervals 
of size~$1$ and green power budget~$0$. 
The goal is to find a solution
of carbon cost~$0$. There is a clear equivalence between the two problems, since the only
way to achieve such a solution is to have a 3-partition of the independent tasks. 
Now, assume there is a  polynomial-time algorithm~$\mathsf{A}$  with 
        $\mathcal{CC}(\mathsf{A}) \leq \lambda \cdot \mathcal{CC}^{*}$.
 Then, if $\mathcal{CC}^{*} = 0$, we have $\mathcal{CC}(\mathsf{A}) = 0$. Otherwise, $\mathcal{CC}^{*} > 0$
 and $\mathcal{CC}(\mathsf{A})\geq \mathcal{CC}^{*}  > 0$. Applying algorithm~$\mathsf{A}$
 on the previous problem instance hence allows us to answer in polynomial time whether there is
 a 3-partition or not, hence proving P=NP. 
\end{proof}
\clearpage
\section{Algorithm Details and Pseudocode}
\label{appendix.algorithm-details}
\subsection{Additional Figures}
\label{appendix.algorithm-details.additional_figures}
\begin{figure}[h]
  \centerline{\includegraphics[width=0.4\linewidth]{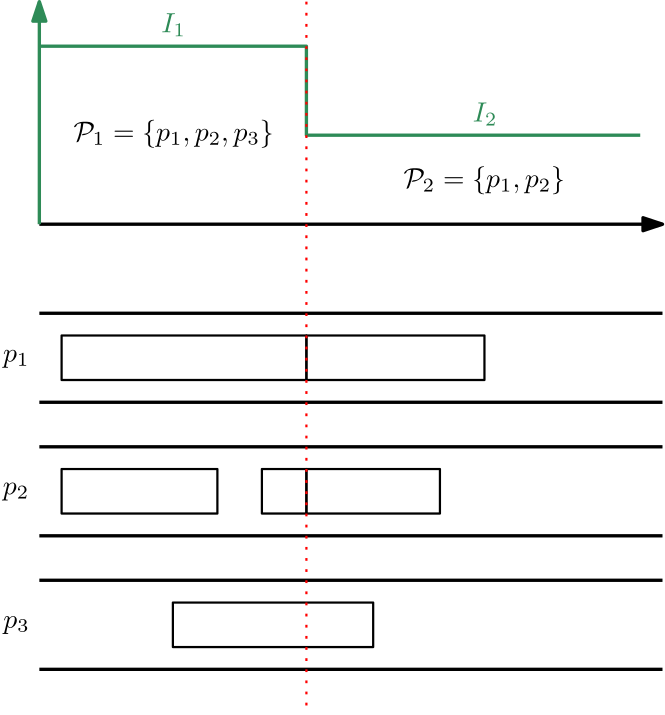}}
  \caption{For interval $I_1$, all processors are allowed to be active, while in the second interval $I_2$, only $p_1$ and $p_2$
  should be active. However, the task on $p_3$ finishes in interval $I_2$. Hence, all three processors are active, 
  potentially violating the given green power budget in $I_2$.}
  \label{fig.dp_tuning_parameter}
\end{figure}
\begin{figure}[h]
    \centerline{\includegraphics[width=0.8\linewidth]{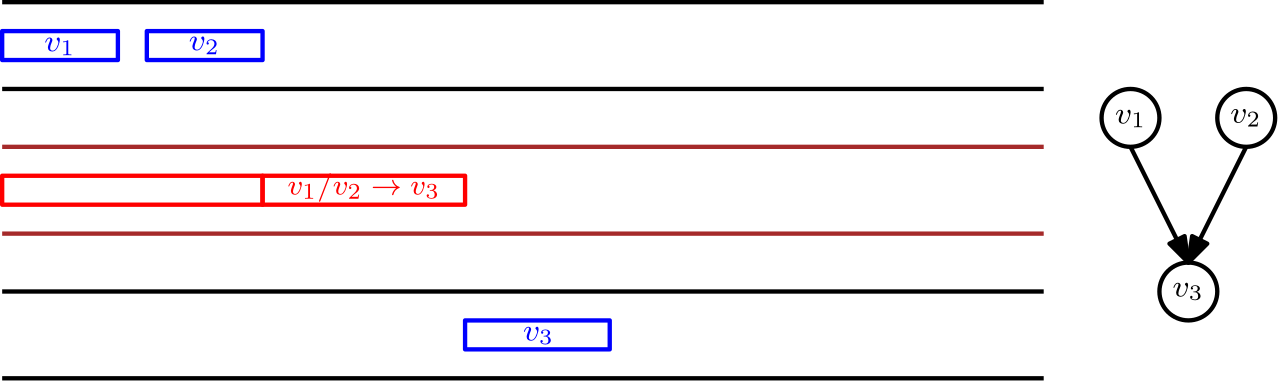}}
    \caption{If we would find for each message the earliest gap on the link and just take the maximum of the finish times, the output of $v_1$
    and $v_2$ would be scheduled at the same time.}
    \label{fig.communication_scheduling}
\end{figure}
\clearpage
\subsection{Local Search Details}
\label{appendix.algorithm-details.local_search_details}
At each iteration, we first compute a set of refined intervals that correspond
to beginning/end of tasks, as well as the carbon cost in each of these intervals
(for details on how to compute the carbon cost efficiently, see~\cite{CaWoSched}, Appendix~A.1). 
Then, we iterate over the intervals from left to right. If we find an interval where the total power consumption exceeds the budget, this becomes 
the interval that we want to resolve. We denote this interval by $I=[b, e)$.
If we do not find one, we end the routine since the carbon cost is zero in this case. 
As already mentioned in Section~\ref{sec.deadline-agnostic-mapping}
and visualized in Figure~\ref{fig.dp_tuning_parameter},
one main reason for exceeding the green power budget for a given interval is that some tasks 
start before the interval but continue running during it, i.\,e., 
tasks $v$ with  
    $\start{v} < b$ and
    $\finish{v} > b$.  
Hence, we look at \emph{all} tasks that run during this interval: 
$\mathbf{C}_I = \{v\in V \ | \  \start{v} < e \ \text{and} \ \finish{v} > b \}.$
We choose one task $v'$ uniformly at random in~$\mathbf{C}_I$, which we aim at shifting 
to a later interval to reduce the power consumption in interval~$I$.
However, we have to ensure that the shift does not violate the 
validity of the schedule. We also consider the tasks {\em following}~$v'$, 
from the task set\new{, i.e., all tasks that start later than $e$ and all successors of $v'$:}
\begin{align*}
    \mathbf{L}'_I & = \{v\in V \ | \  \start{v} \geq e\}  \\
    &\cup \{ v\in V \ |\   v\in\text{succ}(v') 
    \ \text{and} \ \start{v} < e\},
\end{align*} 
where $\text{succ}(v')$ is the set of all successors of $v'$.

Furthermore, we have to ensure that we do not cause overlaps on the processors after shifting tasks.
Hence, we look at every task $v \in \mathbf{L}'_I$ and add all tasks $u$ such that $\mu(u) = \mu(v)$
and $\start{u} \geq \finish{v}$ to a new set~$\mathbf{L}_I$. 
\iftoggle{showOld}{\old{Before we proceed, we perform $\mathbf{L}_I := \mathbf{L}_I \cup \mathbf{L}'_I$.}}{}
\new{We then set $\mathbf{L}_I := \mathbf{L}_I \cup \mathbf{L}'_I$.}
Afterwards, we determine the time units by which we shift tasks. Hence, we look at the minimum move that 
shifts $v'$ outside of the interval $I$, i.\,e.,
\begin{equation*}
    s_{v'} = e - \start{v'}.
\end{equation*}
We also have the possibility to incorporate the given deadline~$D$ into the local search routine.
This becomes important when we refine the solution after the deadline repair step. 
If this is the case, we first compute the latest ending task in the schedule and denote it by~$v_l$. 
Then, we adapt the move to be 
\begin{equation*}
s_{v'} = \min\{(e - \start{v'}), (D-(\finish{v_l}))\}.
\end{equation*}
To be computationally efficient, the idea is now to shift the whole set of tasks $\{v'\} \cup \mathbf{L}_I$ by the same move $s_{v'}$.
Additionally, we have to shift the corresponding communications as well. Hence, we look at each 
outgoing message of each task in $\{v'\} \cup \mathbf{L}_I$ and shift the communication by $s_{v'}$, too. 
This way, we avoid an expensive rescheduling of each communication individually.
The correctness of this approach is shown through the following Lemma~\ref{lemma.local-search-correctness}
\begin{lemma}
\label{lemma.local-search-correctness}
Let $s\in \mathbb{N}$ and let $\sigma$ be a valid schedule (see Section~\ref{sec.framework}).
Furthermore, let $v' \in V$ and $\mathbf{L}_I$
be defined as above. Shifting all tasks in $\{v'\} \ \cup \ \mathbf{L}_I$ and the corresponding outgoing communications by $s_{v'}$
transforms $\sigma$ into another valid schedule $\sigma'$.
\end{lemma}
\begin{proof}
To prove the correctness of a schedule $\sigma$, we have to show the following properties:
\begin{enumerate}
    \item\label{proof.cond1} For each $v \in V$ and each predecessor $u$ of $v$ such that $\mu(u) = \mu(v)$ we have \begin{equation*}
        \finish{u} \leq \start{v}. 
    \end{equation*}
    \item\label{proof.cond2} For each $v \in V$ and each predecessor $u$ of $v$ such that $\mu(u) \neq \mu(v)$ we have \begin{equation*}
        \finish{u}\leq \start{c_{u,v}} \leq \finish{c_{u,v}} \leq \start{v}. 
    \end{equation*}
    \item\label{proof.cond3} There are no overlaps on any processor, i.\,e., there is no pair of tasks $u,v$ such that $\mu(u) = \mu(v)$ and \begin{equation}\label{eq.overlap_case1}
        \start{u} \leq \finish{v} \text{ and } \start{v} \leq \start{u}
    \end{equation} or \begin{equation}\label{eq.overlap_case2}
        \start{u} \leq \start{v} \leq \finish{v} \leq \finish{u}.
    \end{equation}
    \item\label{proof.cond4} Analogously to Property~\ref{proof.cond3}, there are no overlaps on any communication link. 
\end{enumerate}
Now assume $\sigma$, $\sigma'$ as stated in Lemma~\ref{lemma.local-search-correctness}.
We assume that the interval which the local search picked is $I=[b,e[$.
First, we show Property~\ref{proof.cond1}). Let $v\in V$. If $\text{deg}_{in}(v) = 0$, the task 
is independent and hence,~\ref{proof.cond1}) holds. 
Now assume $\text{deg}_{in}(v) \geq 1$ and let $u\in \text{N}_{in}(v)$ and $\mu(u)= \mu(v)$,
i.\,e., the tasks are located on the same processor. 
In $\sigma$, we have $\finish{u} \leq \start{v}$. 
Since tasks never get scheduled earlier but only later, we only have to consider the case when $u$ gets shifted. 
In this case, $u$ is a successor of the task $v'$ picked by the local search --
or it begins later than $e$. By design of the local search, $v$ is moved as well in both cases. 
Hence, we still have $\sigma'(u) + t(u) \leq \sigma'(v)$, and Property~\ref{proof.cond1}) still holds for $\sigma'$.

Now we look at Property~(\ref{proof.cond2}), i.\,e., we consider $u,v\in V$ with $\mu(u) \neq \mu(v)$. 
In this case, there is a message $c_{u,v}$ from $\mu(u)$ to $\mu(v)$. Since $\sigma$ is a valid schedule, 
we have \begin{equation}\label{eq.precedence}
    \finish{u}\leq \start{c_{u,v}} \leq \finish{c_{u,v}} \leq \start{v}. 
\end{equation}
First, note that the message $c_{u,v}$ cannot be moved if $u$ is not moved, by design of the local search.
Hence, with the same arguments as above, we only have to look at the case when $u$ is moved.
In this case, $c_{u,v}$ is moved as well, since all outgoing messages of the shifted messages are shifted as well.
Additionally, by the same arguments as above, because $v$ is a successor of $u$ and starts before $e$ or because $v$ starts later than $e$,
$v$ is moved as well by the same move as $u$ and $c_{u,v}$. Hence, Eq.~(\ref{eq.precedence}) becomes
\begin{align*}
    \sigma'(u) + t(u) + s_{v'} &\leq \sigma'(c_{u,v}) + s_{v'}\\
     &\leq \sigma'(c_{u,v}) + t(c_{u,v}) + s_{v'} \\ 
     &\leq \sigma'(v) + s_{v'}, 
\end{align*} 
and the schedule $\sigma'$ remains correct in this case.

Next, we show Property~\ref{proof.cond3}). Let $u,v\in V$ with $\mu(u) = \mu(v)$ be as in Eq.~(\ref{eq.overlap_case1}). 
First, since every tasks is moved by the same amount, we can assume that only one task is moved. 
Without loss of generality, we assume that in $\cal S$ we have $\finish{u} \leq \start{v}$, i.e., $u$ is processed before $v$. 
First, assume only $u$ is moved. In this case, by design of the local search, all later tasks on $\mu(u)$ are moved as well. 
Hence, this cannot be the case. 
Since in $\sigma$ tasks are only moved to later time units, and because $u$ is processed before $v$ in $\sigma$, we do not have to look at the case 
where only $v$ is moved. Hence, our assumption that there are tasks $u,v\in V$ that overlap on the same processor is a contradiction. 
Further, we see that the same arguments hold for $u,v\in V$ such that Eq.~(\ref{eq.overlap_case2}) holds. 
Hence, Property~(\ref{proof.cond3}) holds. Moreover, since the routines for adding later tasks on a processor 
or communication link is the same, it follows that Property~\ref{proof.cond4}) holds as well by the same arguments. 
Overall, $\sigma'$ is a valid schedule. 
\end{proof}
\clearpage
\subsection{Pseudocode}
\label{appendix.algorithm-details.further_pseudocode}
%
  \begin{algorithm}[h]
  \caption{Processor Subset Selection for Interval $I_j$}
  \label{alg:processor-subset}
  \begin{algorithmic}[1]
  \Require Green power budget $\mathcal{G}_j$, base power
  $\mathcal{P}_\text{base}$, processor speeds $s(\cdot)$, dynamic powers $\workP{}$,
  parameter $\tau$
  \Ensure Processor subset $\mathcal{P}_j$
  \State $C \gets \max\!\bigl\{0,\;\tau\,(\mathcal{G}_j -
  \mathcal{P}_\text{base})\bigr\}$
  \If{$C = 0$} 
    \State \Return $\{p_{\min}\}$ where $p_{\min} = \arg\min_p \workP{p}$ \Comment{tie-breaking by index}
  \EndIf
  \State Solve 0/1 knapsack by dynamic programming (DP), where \\
  items = processors, \\
  weight $= \workP{p}$, \\
  value $= s(p)$, \\
  capacity $= \mathit{C}$
  \State $\mathcal{P}_j \gets$ Reconstruct solution by backtracking through DP table
  \If{$\mathcal{P}_j = \emptyset$}
    \State \Return $\{p_{\min}\}$
  \EndIf 
  \State \Return $\mathcal{P}_j$
  \end{algorithmic}
  \end{algorithm}
  %
  %
  %
  %
  %
  \begin{algorithm}[h]
  \caption{Deadline-Agnostic Initial Mapping / Schedule}
  \label{alg:initial-mapping}
  \begin{algorithmic}[1]
  \Require DAG $G=(V,E,\omega,c)$, intervals $\{I_1, \dots, I_J\}$, tuning parameter $\tau$, processors and communication channels, green power budgets $(\budget_1, \dots, \budget_J)$, $\idleP{}$, $\workP{}$, processor speeds $s(\cdot)$, maximum number of local search iterations $\phi$
  \Ensure Preliminary carbon-aware schedule $\sigma$

  \For{interval $I_j \in \{I_1,\dots,I_J\}$}
    \State $\mathcal{P}_j \gets$ \Call{ProcessorSubset}{$\mathcal{G}_j$, $\mathcal{P}_\text{base}$, $s(\cdot)$, $\workP{}$, $\tau$} \Comment{Algorithm~\ref{alg:processor-subset}}
  \EndFor
  \State Compute $\mathrm{rank}(v) = \overline{\mathrm{rt}(v)} + \max_{(v,w)\in
   E}(c(v,w) + \mathrm{rank}(w))$ for all $v\in V$
  \State $\pi \gets$ Sort $V$ by $\mathrm{rank}(\cdot)$ descending (random tie-breaking)
  \For{task $v \in \pi$}
    \State $\mathrm{EST}'(v) \gets \max_{(u,v)\in E}\finish{u}$
    \State Find $I_j = [b_j, e_j)$ s.t.\ $b_j \le \mathrm{EST}'(v) < e_j$
    \State $(p_{best}, \sigma(v)') \gets$
      \Call{FindChoice}{$v,\,\mathcal{P}_j$}
      \Comment{pick $\arg\min$ EFT}
    \State $\mathit{retries} \gets 0$
    \While{$\start{v}' \notin I_j$ \textbf{and} $\mathit{retries} < 3$
  \textbf{and} $j{+}1 < J$}
      \State $j \gets j+1$;\quad $\mathit{retries} \gets \mathit{retries}+1$
      \State $(p_{best}, \sigma(v)') \gets$
        \Call{FindChoice}{$v,\,\mathcal{P}_j,\,\text{start not earlier than } b_j$}
    \EndWhile
    \State $\mu(v) \gets p_{best};\ \start{v} \gets \start{v}';\ \text{schedule all communications } c_{u,v} \text{ insertion based}$ 
  \EndFor
  \State Apply \Call{LocalSearch}{$\sigma$, $D{=}\infty$, $\phi$} \Comment{Algorithm~\ref{alg:local-search}}
  \State \Return $\sigma$
  \end{algorithmic}
  \end{algorithm}
  %
  %
  %
  %
  %
  \begin{algorithm}[h]
  \caption{Local Search}
  \label{alg:local-search}
  \begin{algorithmic}[1]
  \Require Schedule $\sigma$, deadline $D$, max iterations $\phi$, intervals $\{I_1,\dots,I_J\}$
  \Ensure Refined schedule $\sigma$

  \For{$i = 1$ \textbf{to} $\phi$}
    \State Compute refined intervals by sweep-line over task/communications and interval bounds
    \State Find leftmost refined interval $I=[b,e)$ with total power $>\mathcal{G}_I$;
    \If{no such interval exists}
      \State \Return
    \EndIf
    \State $\mathbf{C}_I \gets \{v \in V \mid \start{v} < e \text{ and }
  \finish{v} > b\}$
      \Comment{tasks running in $I$}
    \If{$\mathbf{C}_I = \emptyset$} 
      \State \Return
    \EndIf
    \State $v' \gets$ draw task uniformly at random from $\mathbf{C}_I$
    \State $\mathbf{L}'_I \gets \{v \in V \mid \start{v} \ge e\}
             \;\cup\; \{u \in \mathrm{succ}(v') \mid \start{u} < e\}$
    \For{$v \in \mathbf{L}_I'$}
      \State $\mathbf{L}_I \gets \mathbf{L}_I \cup \{u \in V \mid \mu(u) = \mu(v) \text{ and } \start{u} \geq \finish{v}\}$
    \EndFor
    \State $\mathbf{L}_I \gets \mathbf{L}'_I$
    \State $v_l \gets \arg\max_{u \in \{v'\}\cup\mathbf{L}_I} \finish{u}$
    \State $s_{v'} \gets \min\!\bigl(e - \start{v'},\; D - (\finish{v_l})\bigr)$
    \If{$s_{v'} \le 0$}
    \State \Return    
    \EndIf
    \State Shift all tasks in $\{v'\}\cup\mathbf{L}_I$ by $s_{v'}$
    \State Shift all outgoing communications of tasks in $\{v'\}\cup\mathbf{L}_I$ by $s_{v'}$
  \EndFor
  \State \Return $\sigma$
  \end{algorithmic}
  \end{algorithm}
  %
  %
  %
  %
  %
  \begin{algorithm}[h]
  \caption{Deadline Repair}
  \label{alg:deadline-repair}
  \begin{algorithmic}[1]
  \Require Schedule $\sigma$ (from Algorithm~\ref{alg:initial-mapping}),
  deadline $D$
  \Ensure Deadline-feasible schedule $\sigma'$
  \If{makespan$(\sigma) \le D$}
    \State Apply \Call{LocalSearch}{$\sigma$, $D$}
    \State \Return $\sigma$
  \EndIf
  \State $\pi \gets$ Re-compute upward ranks and sort $V$ by rank descending
  \State $\sigma' \gets $ \Call{Reschedule}{$\sigma$, $\xi{=}D$, $\pi$} \Comment{Algorithm~\ref{alg:rescheduling}}
  \If{makespan$(\sigma') \le D$}
    \State Apply \Call{LocalSearch}{$\sigma'$, $D$}
    \State \Return $\sigma'$
  \EndIf
  \State $\xi_\text{low} \gets 0$;\quad $\xi_\text{high} \gets D$
  \While{$\xi_\text{low} + 1 < \xi_\text{high}$}
    \State $\xi_\text{mid} \gets \xi_\text{low} +
  \lfloor(\xi_\text{high}-\xi_\text{low})/2\rfloor$
    \State $\sigma' \gets$ \Call{Reschedule}{$\sigma$, $\xi_\text{mid}$, $\pi$} \Comment{Algorithm~\ref{alg:rescheduling}}
    \If{makespan$(\sigma') \le D$}
      \State $\xi_\text{low} \gets \xi_\text{mid}$
    \Else
      \State $\xi_\text{high} \gets \xi_\text{mid}$
    \EndIf
  \EndWhile
  \State $\sigma' \gets$ \Call{Reschedule}{$\sigma$, $\xi_\text{low}$, $\pi$} \Comment{Algorithm~\ref{alg:rescheduling}}
  \State Apply \Call{LocalSearch}{$\sigma'$, $D$}
  \State \Return $\sigma'$
  \end{algorithmic}
  \end{algorithm}
  %
  %
  %
  %
  %
  \begin{algorithm}[h]
  \caption{Rescheduling Procedure for Deadline Repair}
  \label{alg:rescheduling}
  \begin{algorithmic}[1]
  \Require Schedule $\sigma$, deadline $D$, threshold $\xi$
  \Ensure Shifted schedule $\sigma'$
    \State $\mathcal{R}' \gets \{v \in V \mid \finish{v} > \xi\}$
    \State Via DFS:
           $\mathcal{R} \gets \mathcal{R}' \cup \{u \mid \exists\, v\in\mathcal{R}' \text{ with
  path } v\to u\}$
    \State $\sigma' \gets$ Initialize by keeping schedule of $V\setminus\mathcal{R}$
  from $\sigma$ fixed
    \For{$v \in \mathcal{R}$ in order $\pi$}
      \State Schedule $v$ using \heft \ on \emph{all} processors \Comment{no subset restriction}
    \EndFor
    \State \Return $\sigma'$
  \end{algorithmic}
  \end{algorithm}
\clearpage
\section{Setup Details}
\label{appendix.setup-details}
\noindent{\bf Workflows. }
%
In line with Section~\ref{sec.framework}, we represent the workflows as DAGs. 
To create them, we use the same tools and techniques as in~\cite{CaWoSched}:
we first take real-world workflows (atacseq, bacass, methylseq (only a rescaled version is used), eager, chipseq) 
from~\cite{lotaru}, transform them to our model, and scale them with the WFGen generator~\cite{wfcommons} 
to the desired size.  The resulting 44 workflows have between $12$ and $\numprint{30000}$ vertices.  
The DAG weights are sampled 
from a normal distribution, where mean and standard deviation are derived from the speed values of the cluster. 

\smallskip
\noindent{\bf Profiles. }
%
For the power budget profiles, we use a more realistic approach than~\cite{CaWoSched}.
We sample real carbon intensities from an Electricity Maps dataset~\cite{electricitymaps}, using the hourly data for Germany and for California in 2024. 
A given time horizon is first partitioned into intervals by computing random split points.
Thereby, we ensure that each interval length is uniform within a configurable integer range (we chose $[10, 50]$). 
For the intervals, we pick a contiguous subsequence of carbon intensity values uniformly at random. 
Depending on the cluster type, we then rescale the values as follows to invert the carbon
intensities, making the smallest carbon intensity correspond to the largest green power budget:
For the sequence of values, we denote by $[x_{\min}, x_{\max}]$ the range of carbon intensities. 
Then we map this range to a target power interval $[P_{\min}, P_{\max}]$, where we choose 
$P_{\min}$ to be the sum of all idle powers in the cluster and for $P_{\max}$ we add 
a fraction ($0.2$ for the small cluster, $0.4$ for the large cluster) of the sum of all dynamic power values
for the compute nodes and the communication links.
Then, any carbon intensity $x$ becomes 
\begin{equation}
\label{eq:carbon-intensity-power-budget}
    x^{\prime} = P_{\max} - \frac{x-x_{\min}}{x_{\max} - x_{\min}}(P_{\max} - P_{\min}).
\end{equation}
\clearpage
\section{Additional Experimental Results}
\label{appendix.further-results}
In this section, we show additional plots for the experimental evaluation of the algorithms 
performed in Section~\ref{sec.experiments}.
\begin{figure}[h]
\centering 
\begin{subfigure}{0.48\textwidth}
  \centerline{\includegraphics[width=\linewidth]{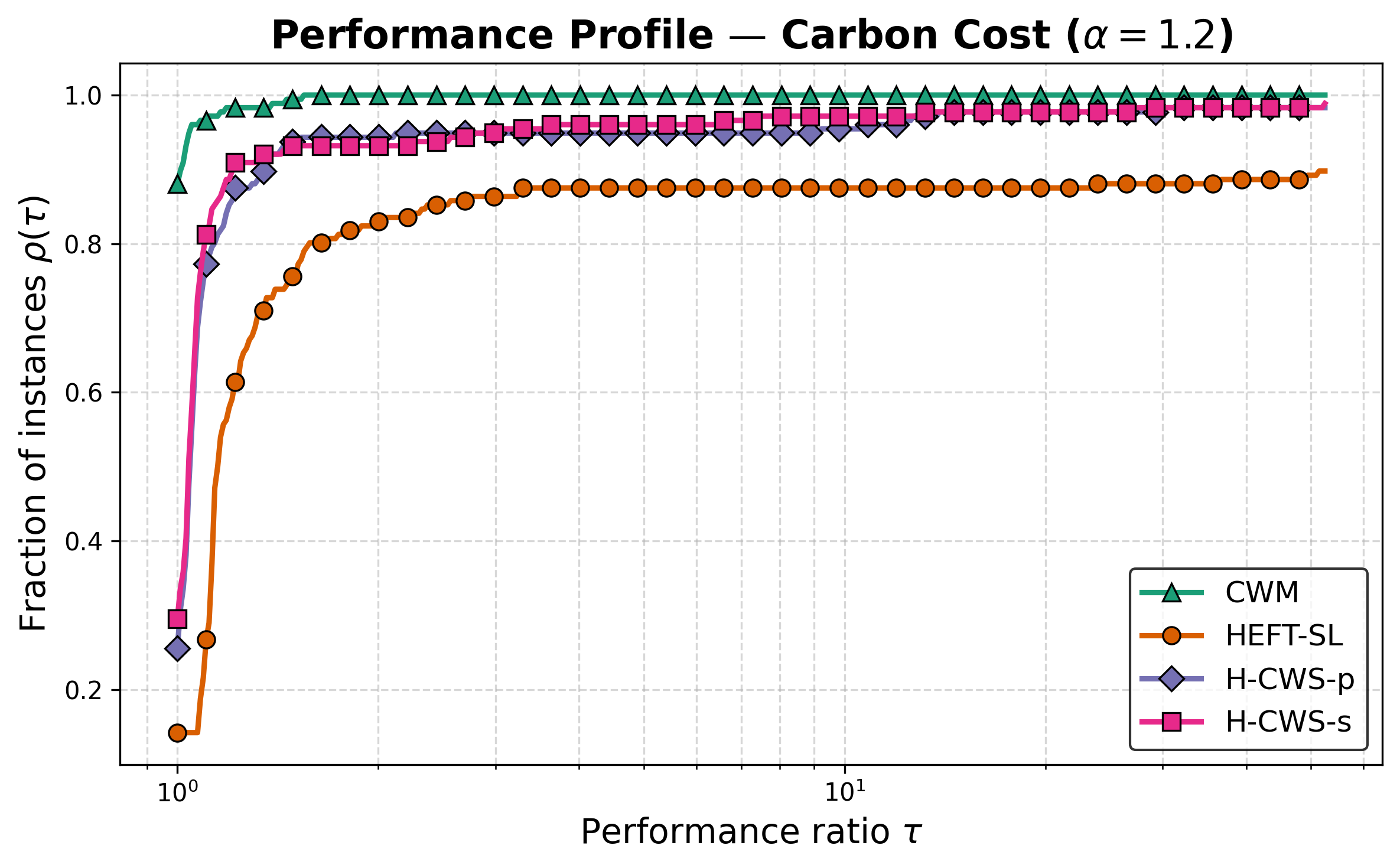}}
\end{subfigure}
\hfill
\begin{subfigure}{0.48\textwidth}
  \centerline{\includegraphics[width=\linewidth]{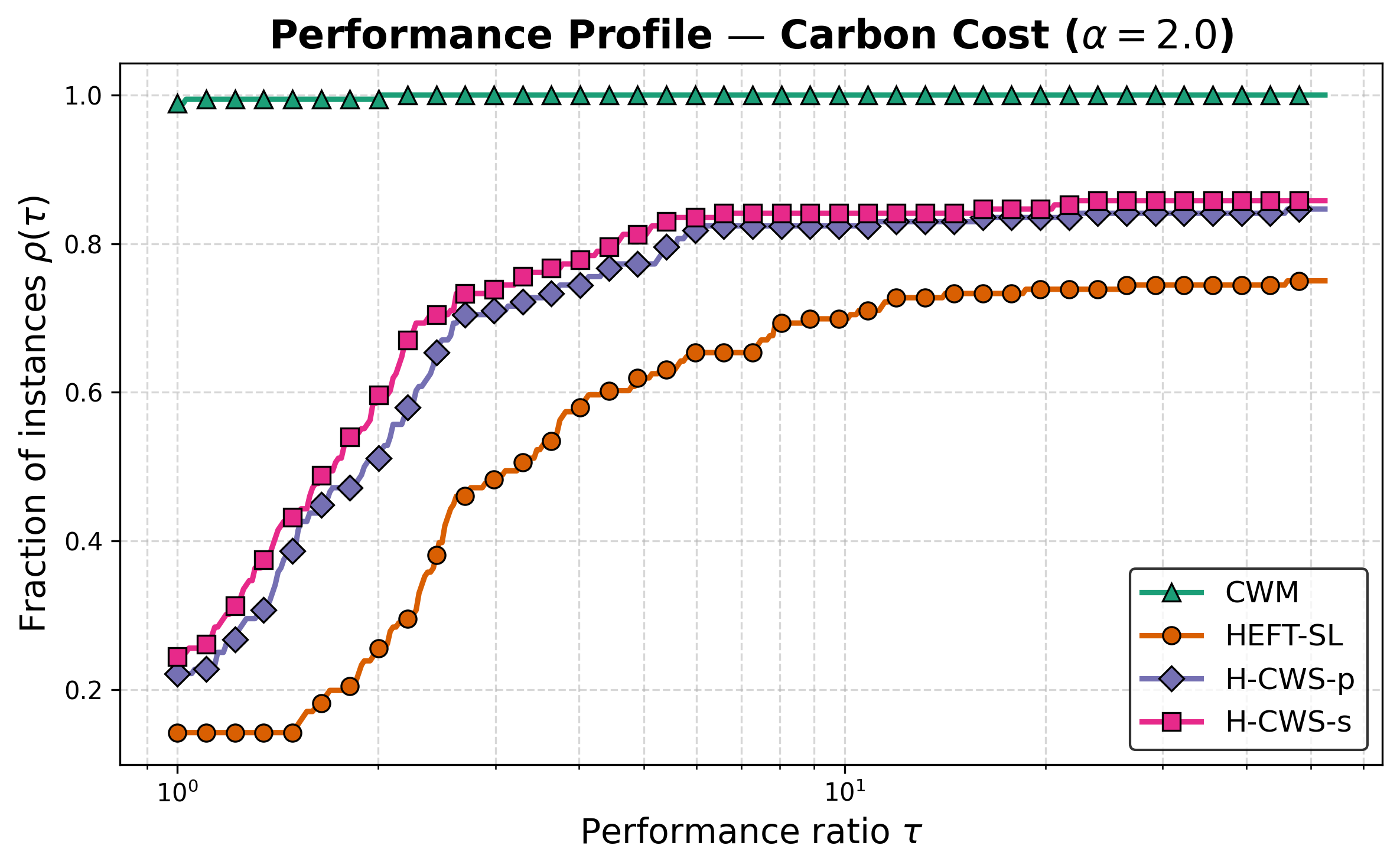}}
\end{subfigure}
\caption{Performance profiles for deadlines 
$D=1.2\times M$ (left) and $D=2.0\times M$ (right).}
\label{fig.performance-profile_12_20}
\end{figure}
\begin{figure}[h]
  \centerline{\includegraphics[width=0.6\textwidth]{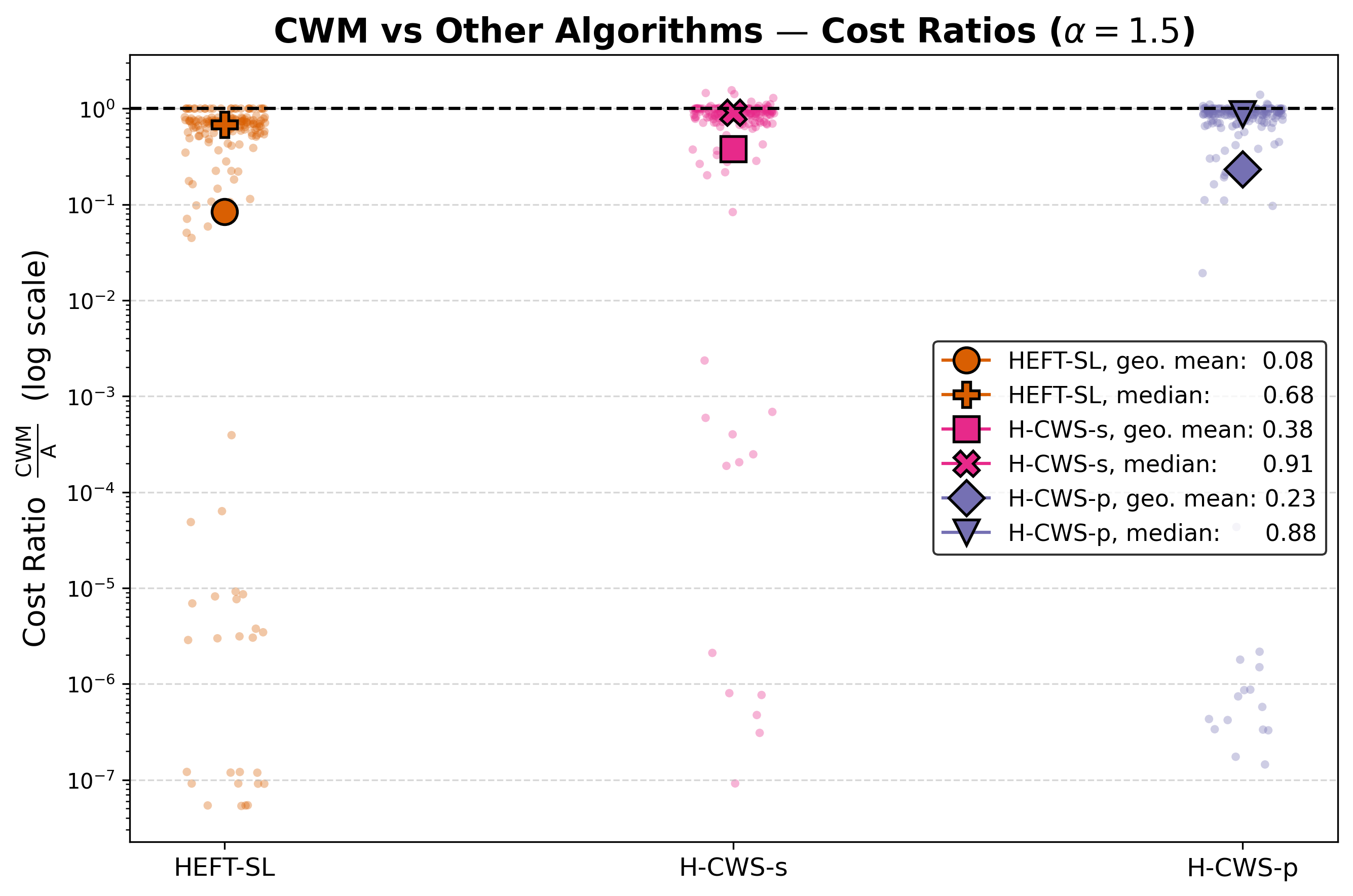}}
  \caption{Carbon cost ratios as described in Section~\ref{par.cost-analysis} for 
  deadline $D=1.5\times M$.}
  \label{fig.cost_ratios_15}
\end{figure}
\begin{figure}[h]
\centering 
\begin{subfigure}{0.48\textwidth}
  \centerline{\includegraphics[width=\linewidth]{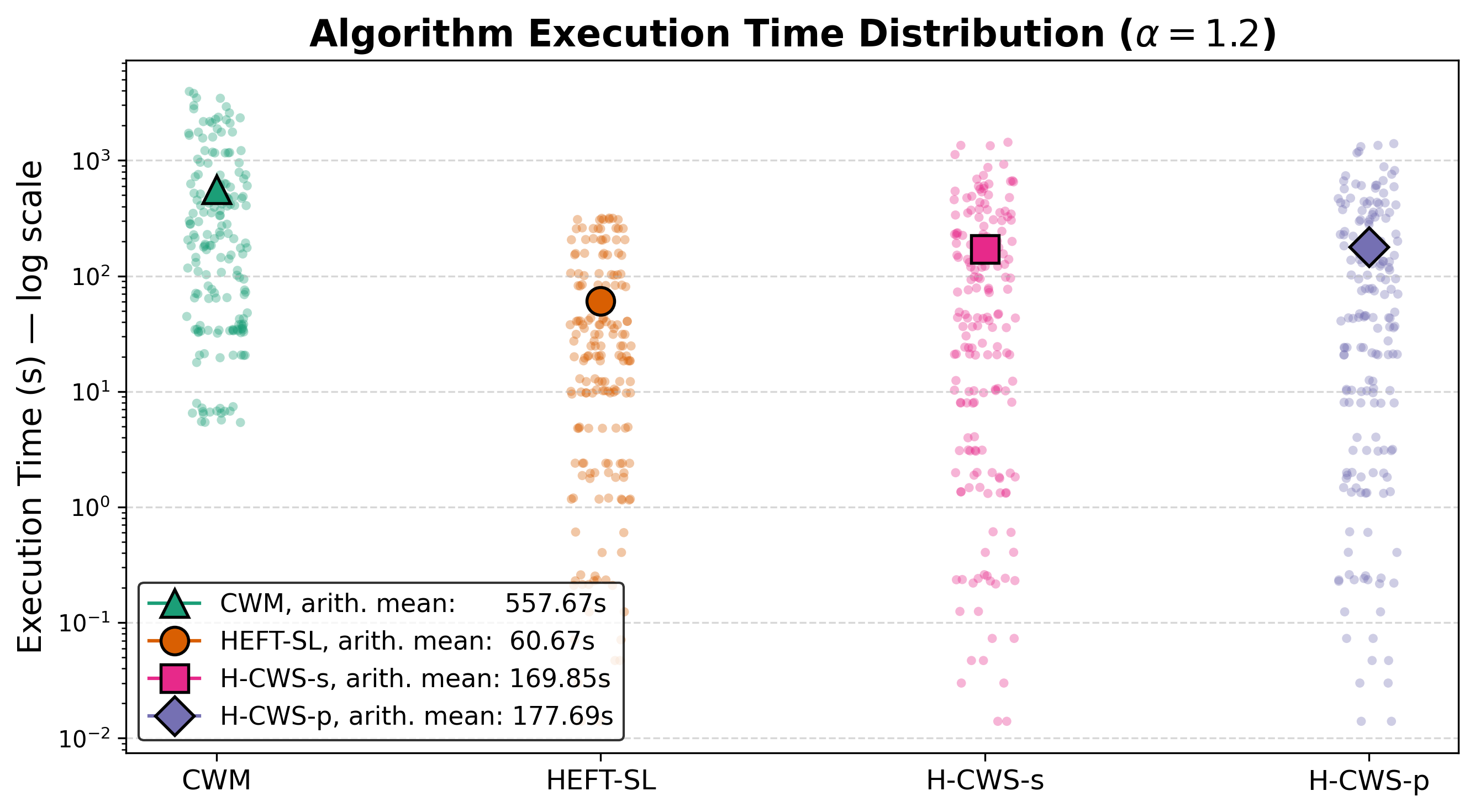}}
\end{subfigure}
\hfill
\begin{subfigure}{0.48\textwidth}
  \centerline{\includegraphics[width=\linewidth]{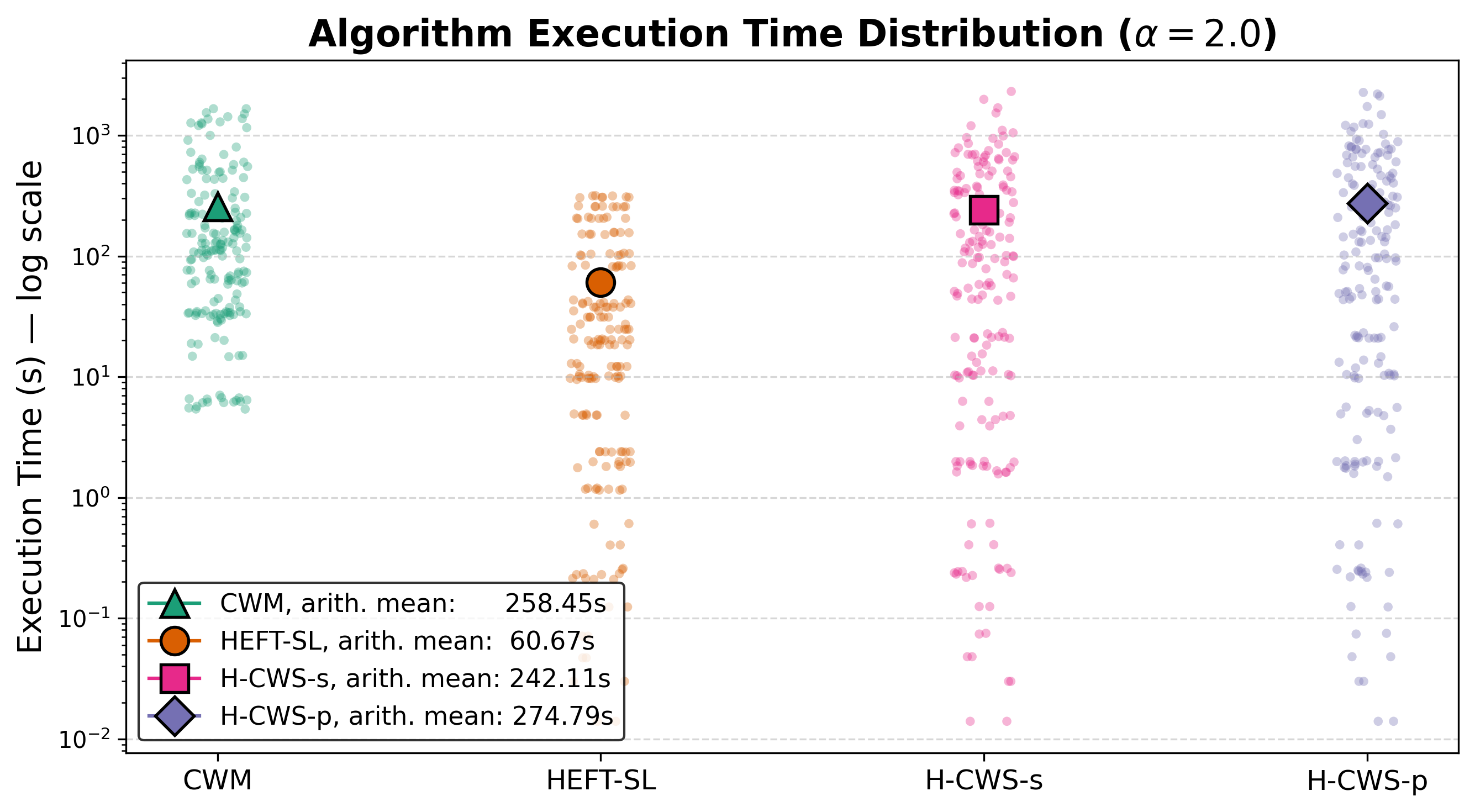}}
\end{subfigure}
\caption{Absolute execution time (seconds) for each algorithm for deadlines 
$D=1.2\times M$ (left) and $D=2.0\times M$ (right).}
\label{fig.abs-time_12_20}
\end{figure}
\begin{figure}[h]
\centering
\begin{subfigure}{0.48\textwidth}
  \centerline{\includegraphics[width=\linewidth]{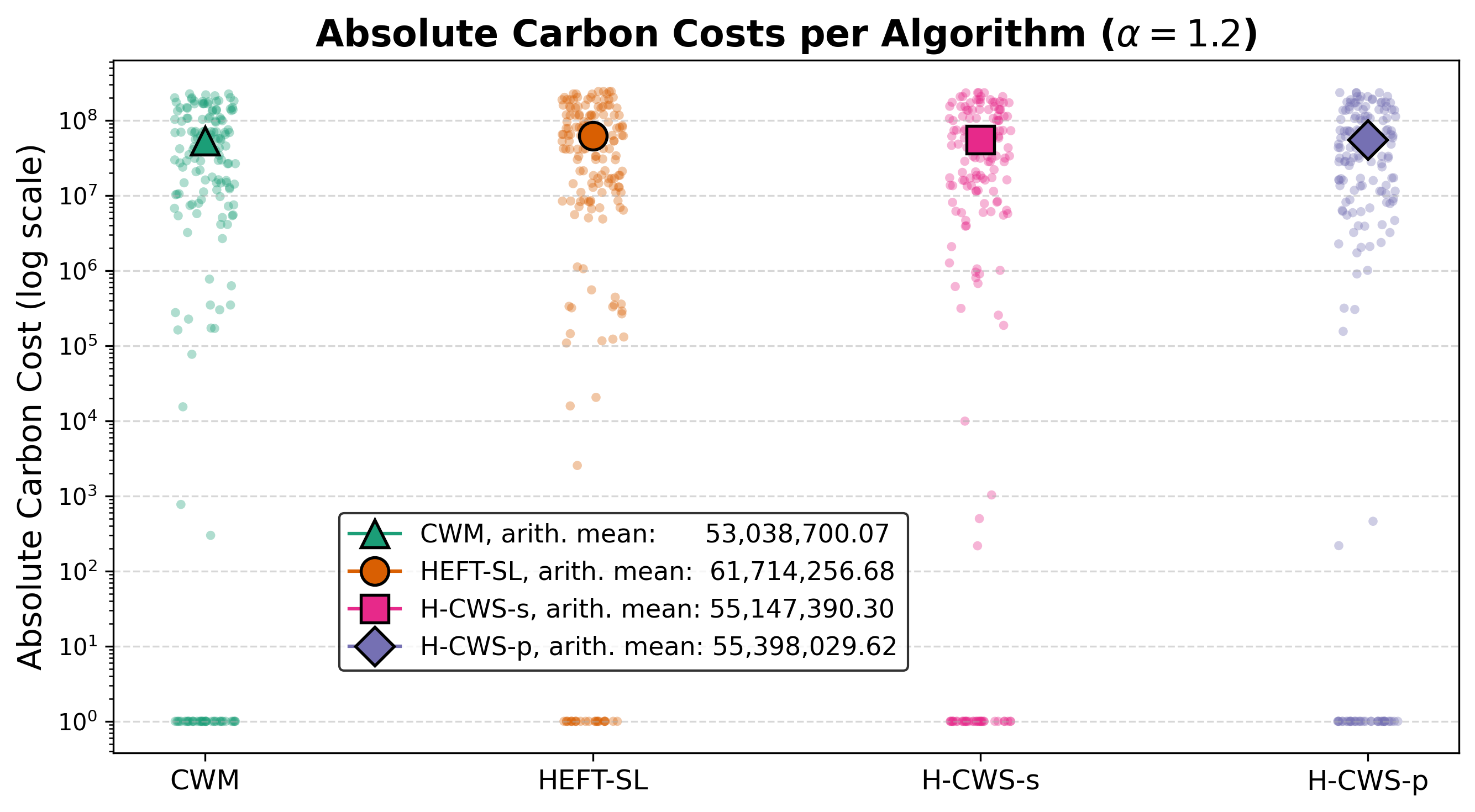}}
\end{subfigure}
\hfill
\begin{subfigure}{0.48\textwidth}
  \centerline{\includegraphics[width=\linewidth]{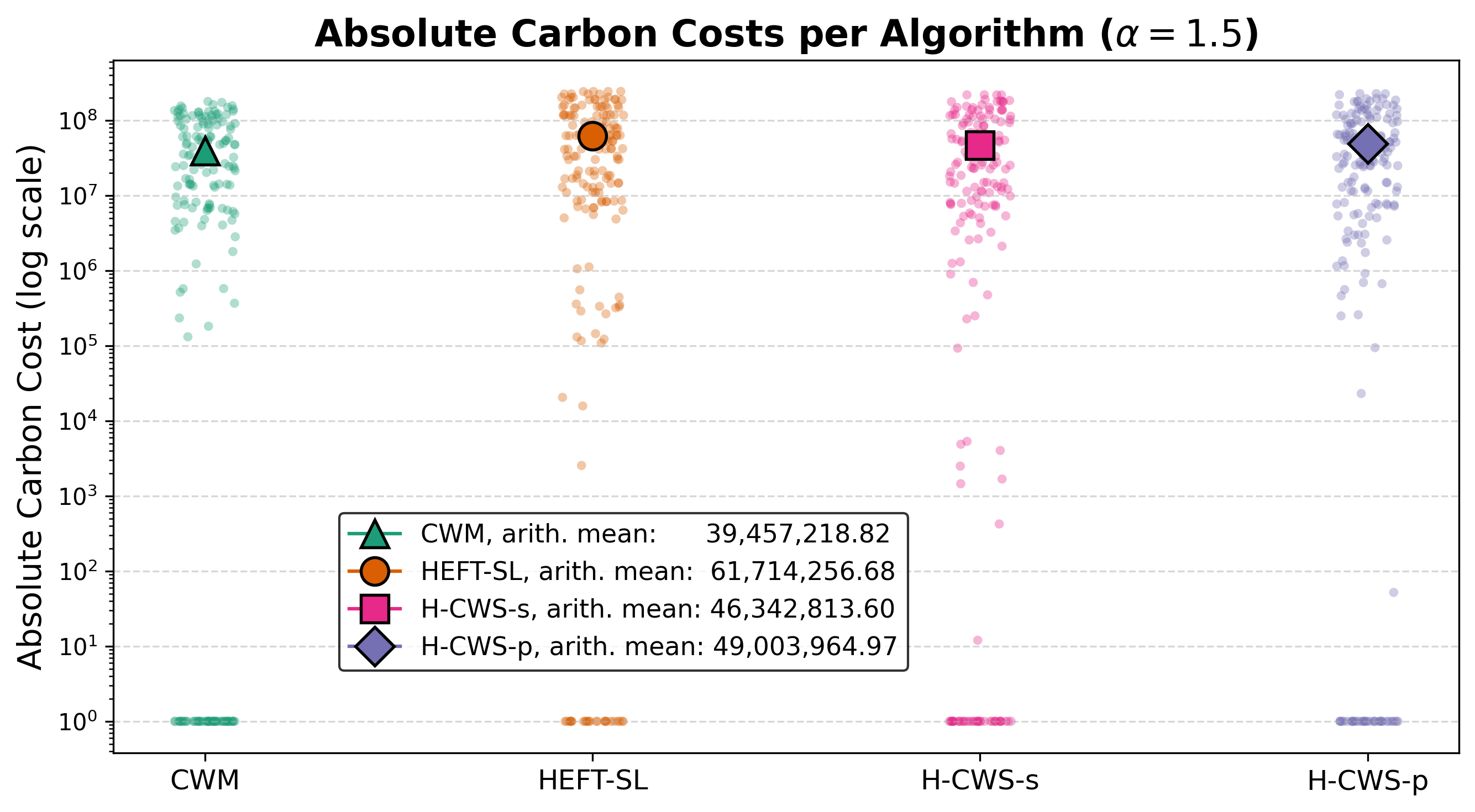}}
\end{subfigure}
\begin{subfigure}{0.48\textwidth}
  \centerline{\includegraphics[width=\linewidth]{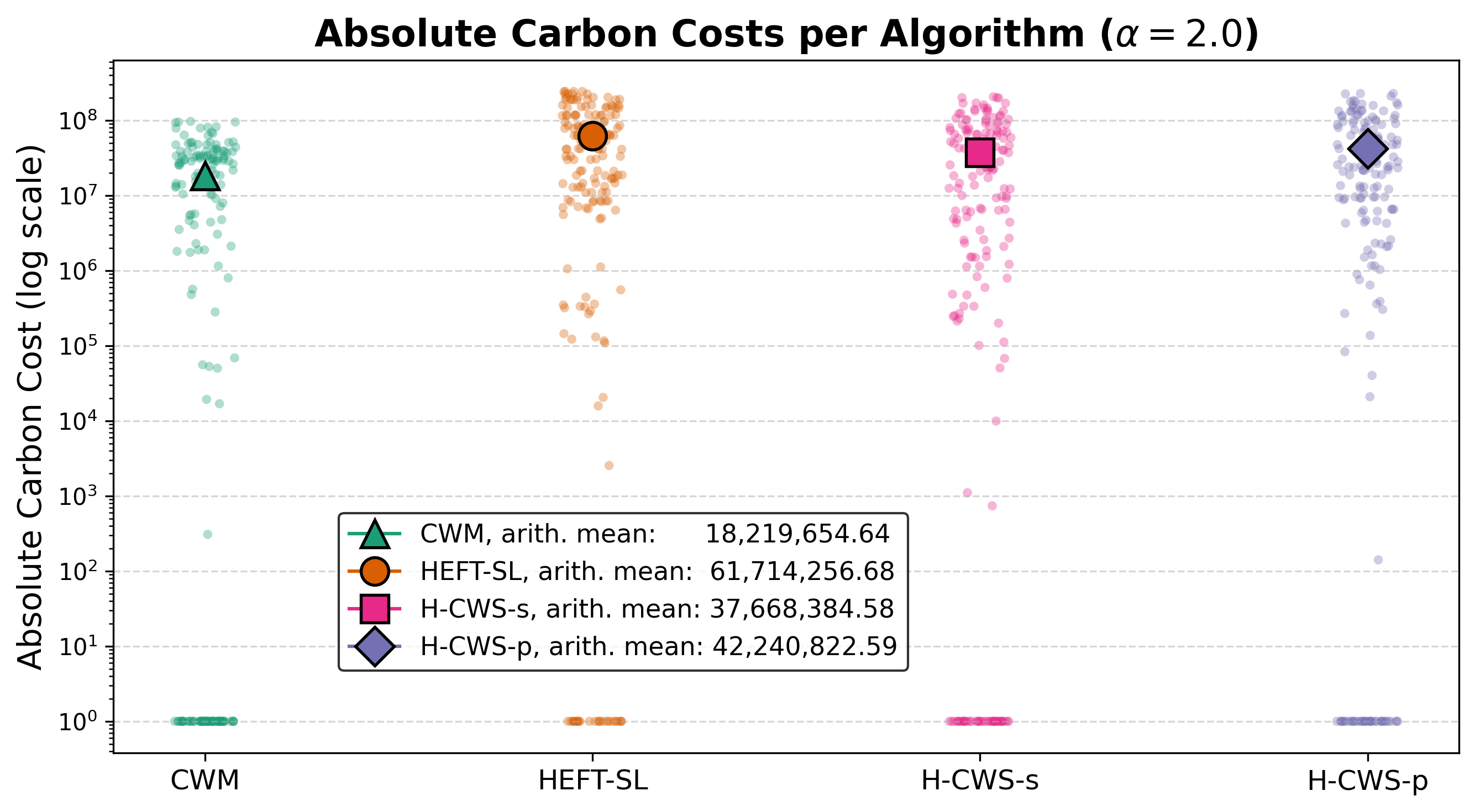}}
\end{subfigure}
\caption{Absolute carbon cost for each algorithm for deadlines $D=1.2\times M$ (top left), $D=1.5\times M$ (top right),
and $D=2.0\times M$ (bottom).}
\label{fig.abs-cost_12_15_20}
\end{figure}
\begin{figure}[h]
\centering
\begin{subfigure}{0.48\textwidth}
  \centerline{\includegraphics[width=\linewidth]{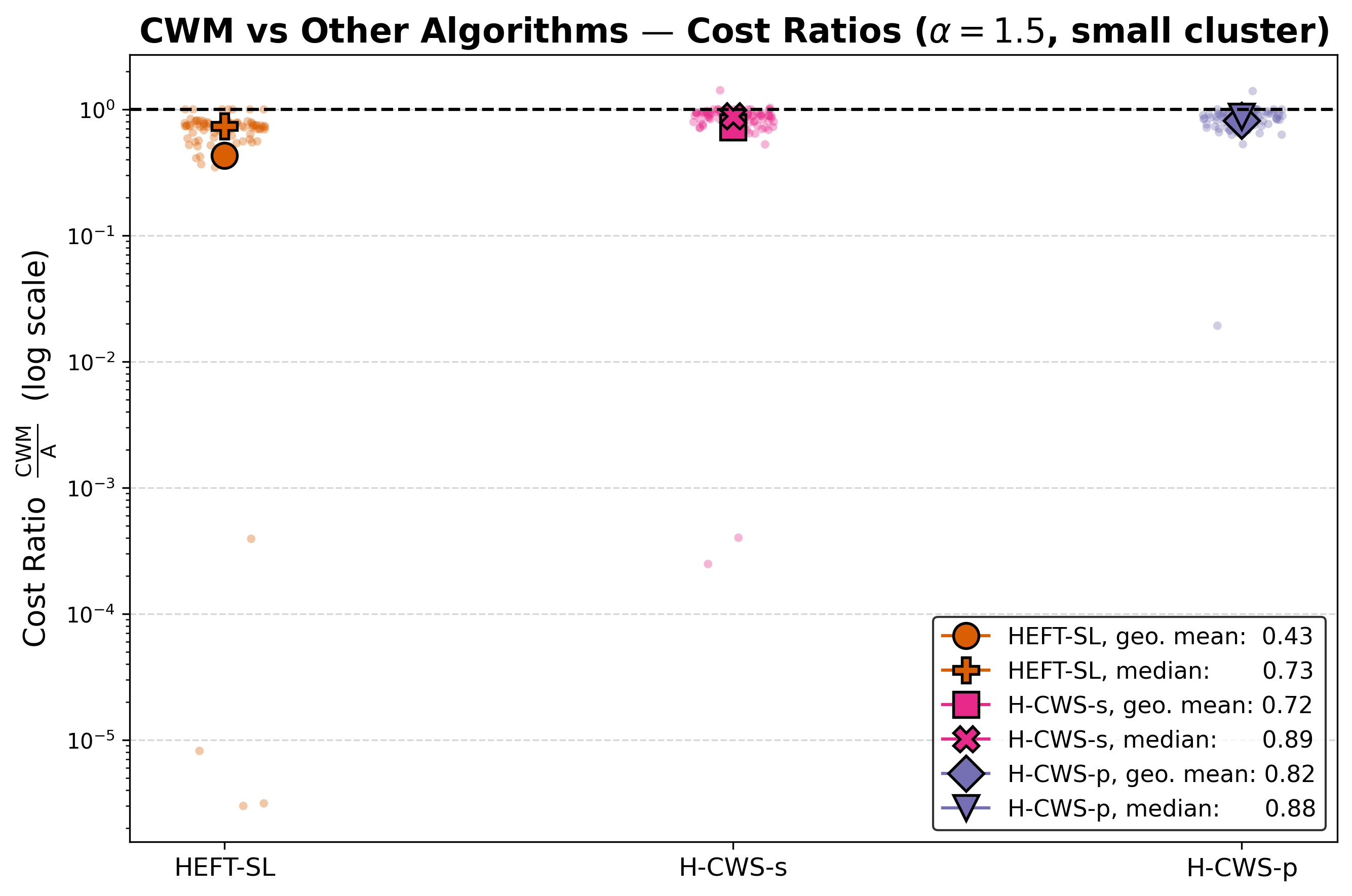}}
\end{subfigure}
\hfill
\begin{subfigure}{0.48\textwidth}
  \centerline{\includegraphics[width=\linewidth]{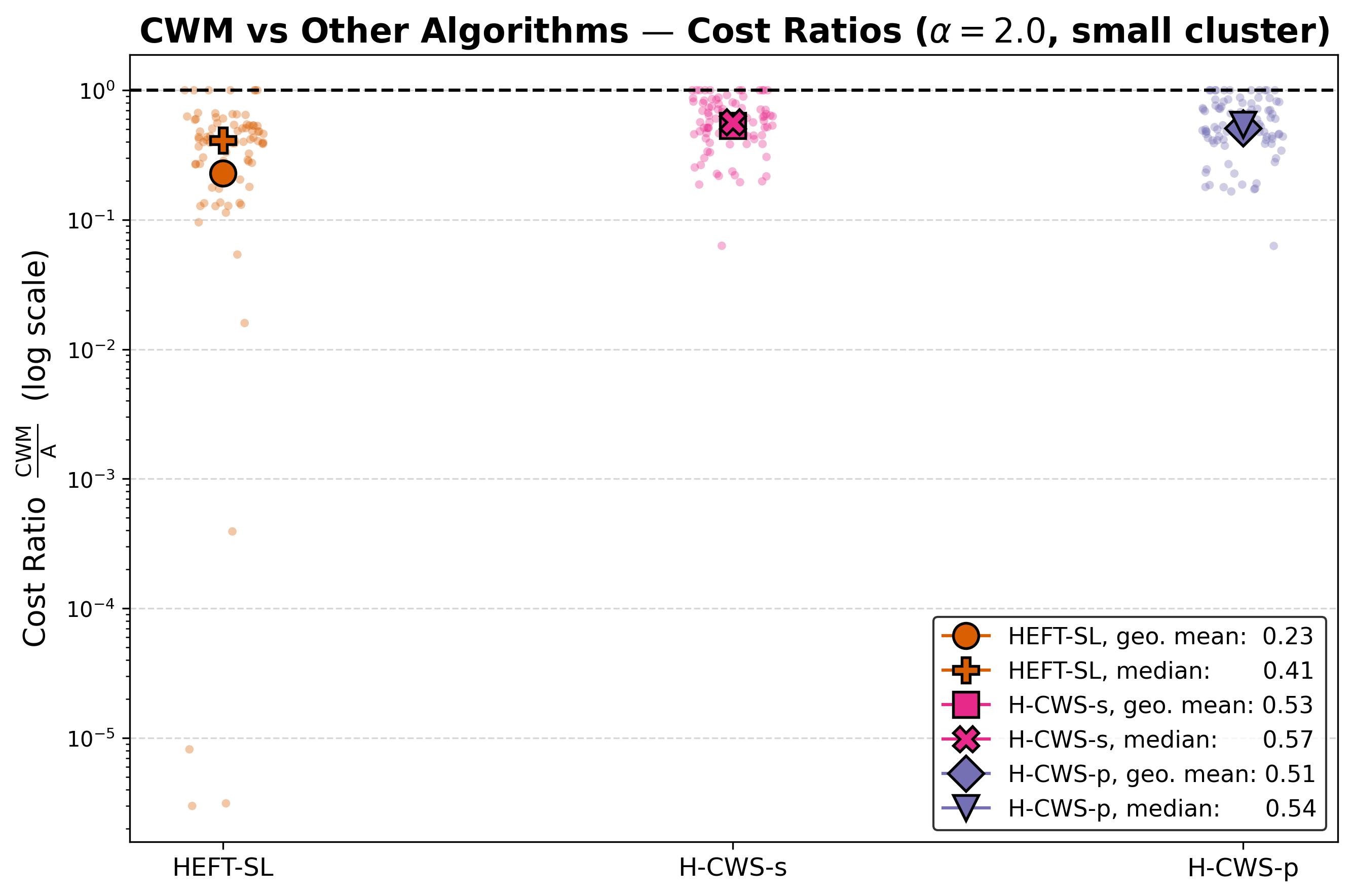}}
\end{subfigure}
\begin{subfigure}{0.48\textwidth}
  \centerline{\includegraphics[width=\linewidth]{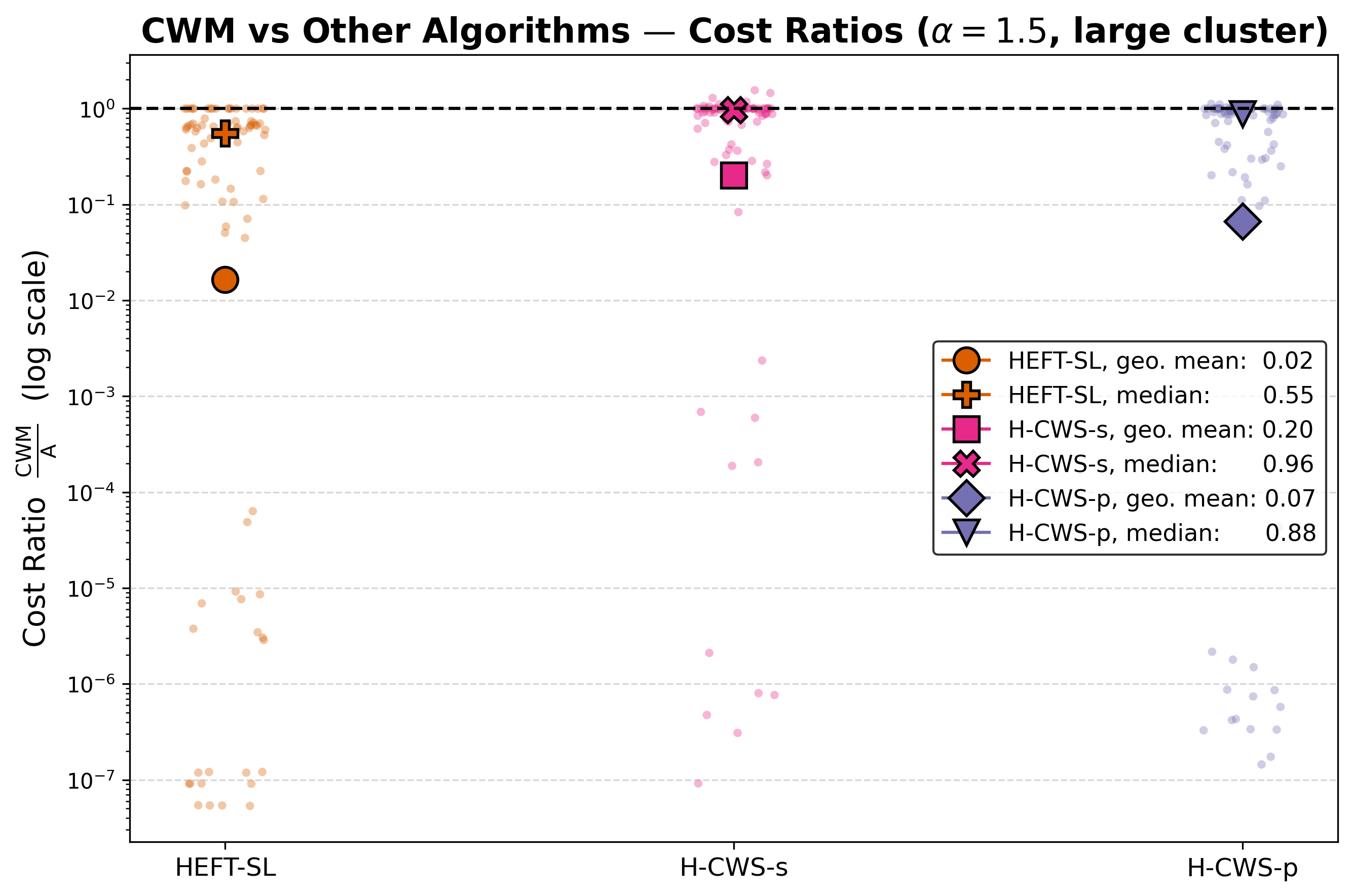}}
\end{subfigure}
\hfill
\begin{subfigure}{0.48\textwidth}
  \centerline{\includegraphics[width=\linewidth]{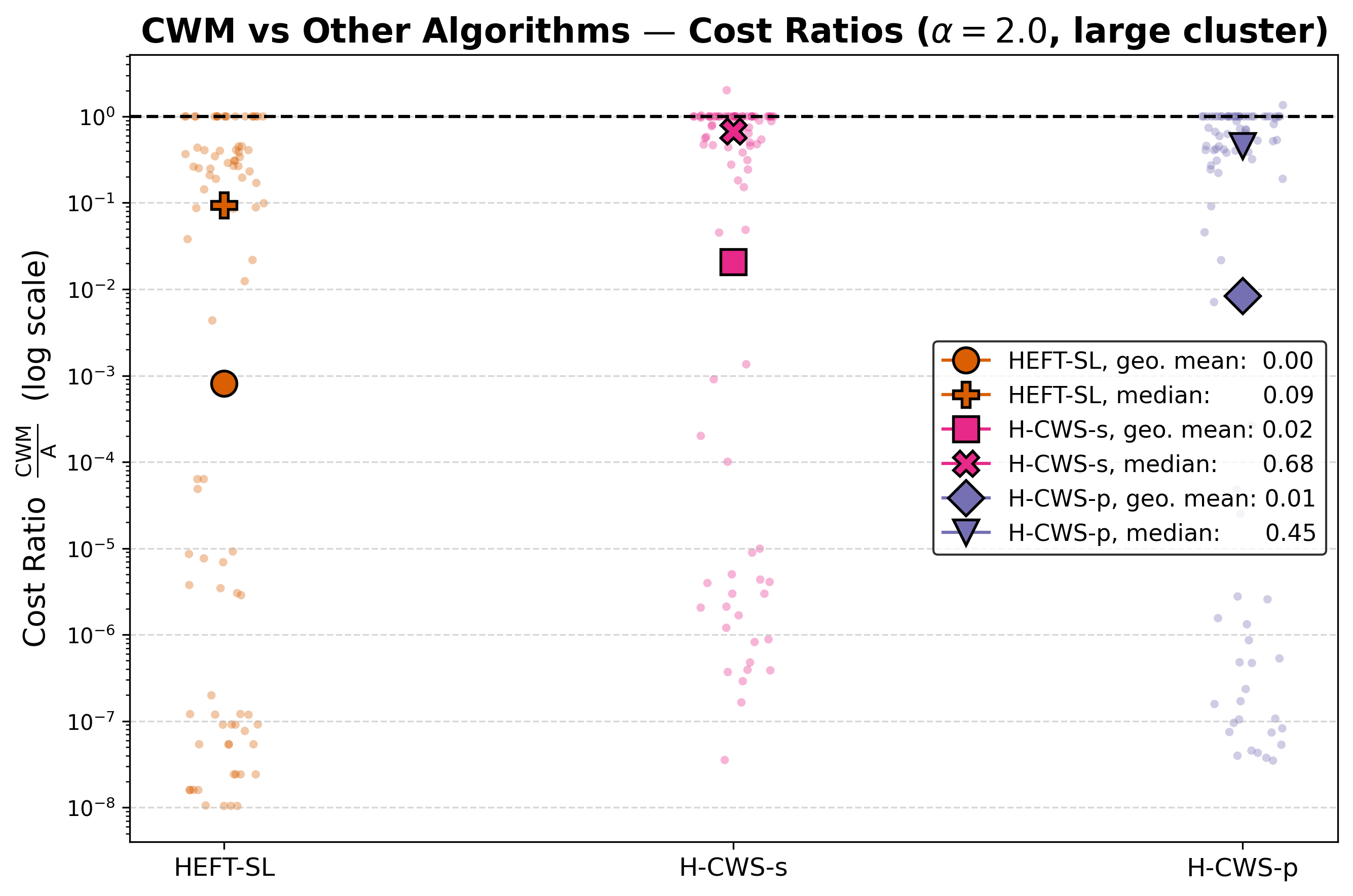}}
\end{subfigure}
\caption{Carbon cost ratios for each competitor for the small and large cluster and for the deadlines
$D=1.5\times M$ and $D=2.0\times M$. Small cluster, $D=1.5\times M$: top left; Small cluster, $D=2.0\times M$: top right;
Large cluster, $D=1.5\times M$: bottom left; Large cluster, $D=2.0\times M$: bottom right.}
\label{fig.cost-ratios_cluster-size}
\end{figure}
\begin{figure}[h]
  \begin{subfigure}{0.48\textwidth}
    \centerline{\includegraphics[width=\linewidth]{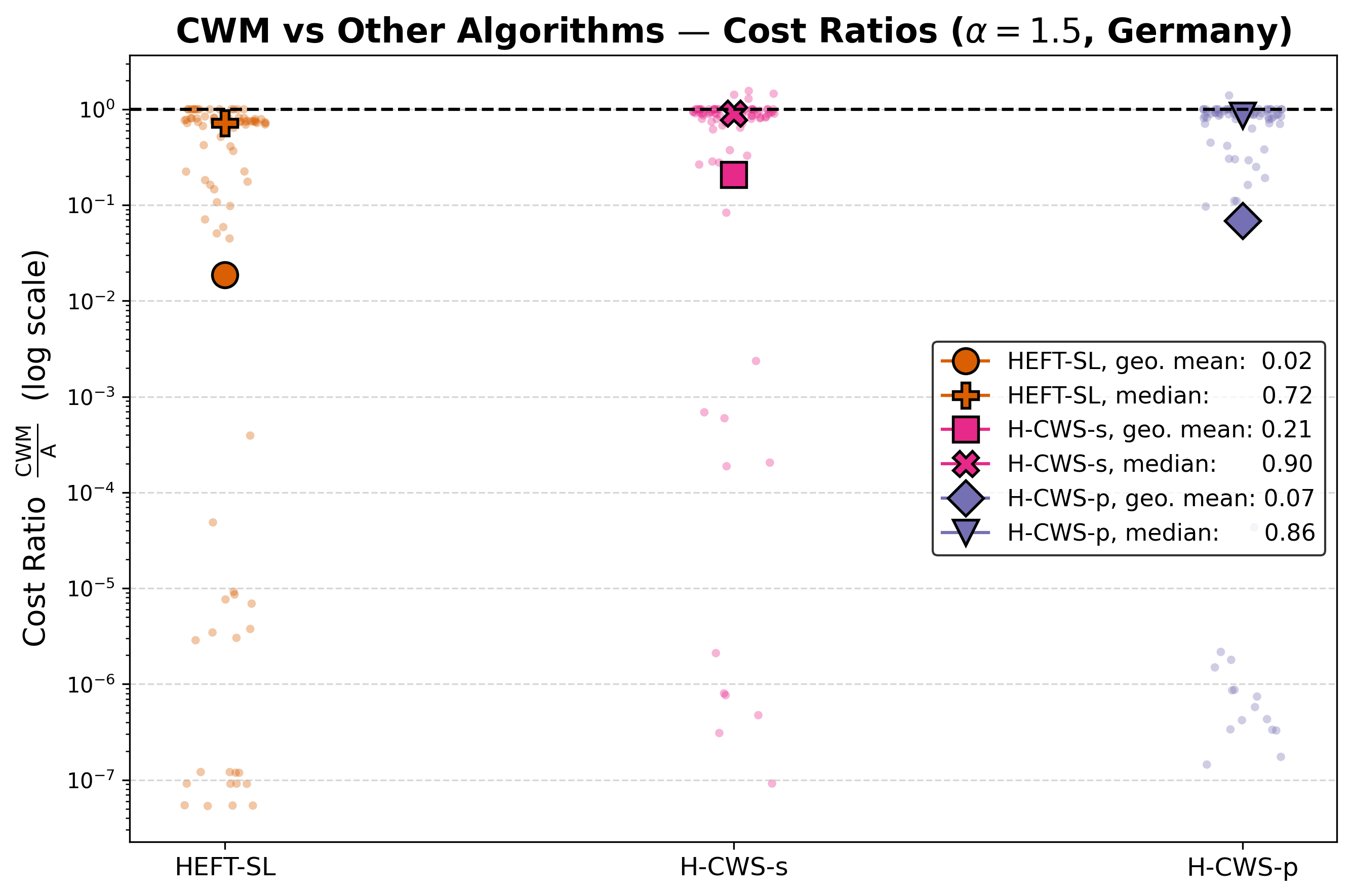}}
  \end{subfigure}
  \hfill
  \begin{subfigure}{0.48\textwidth}
    \centerline{\includegraphics[width=\linewidth]{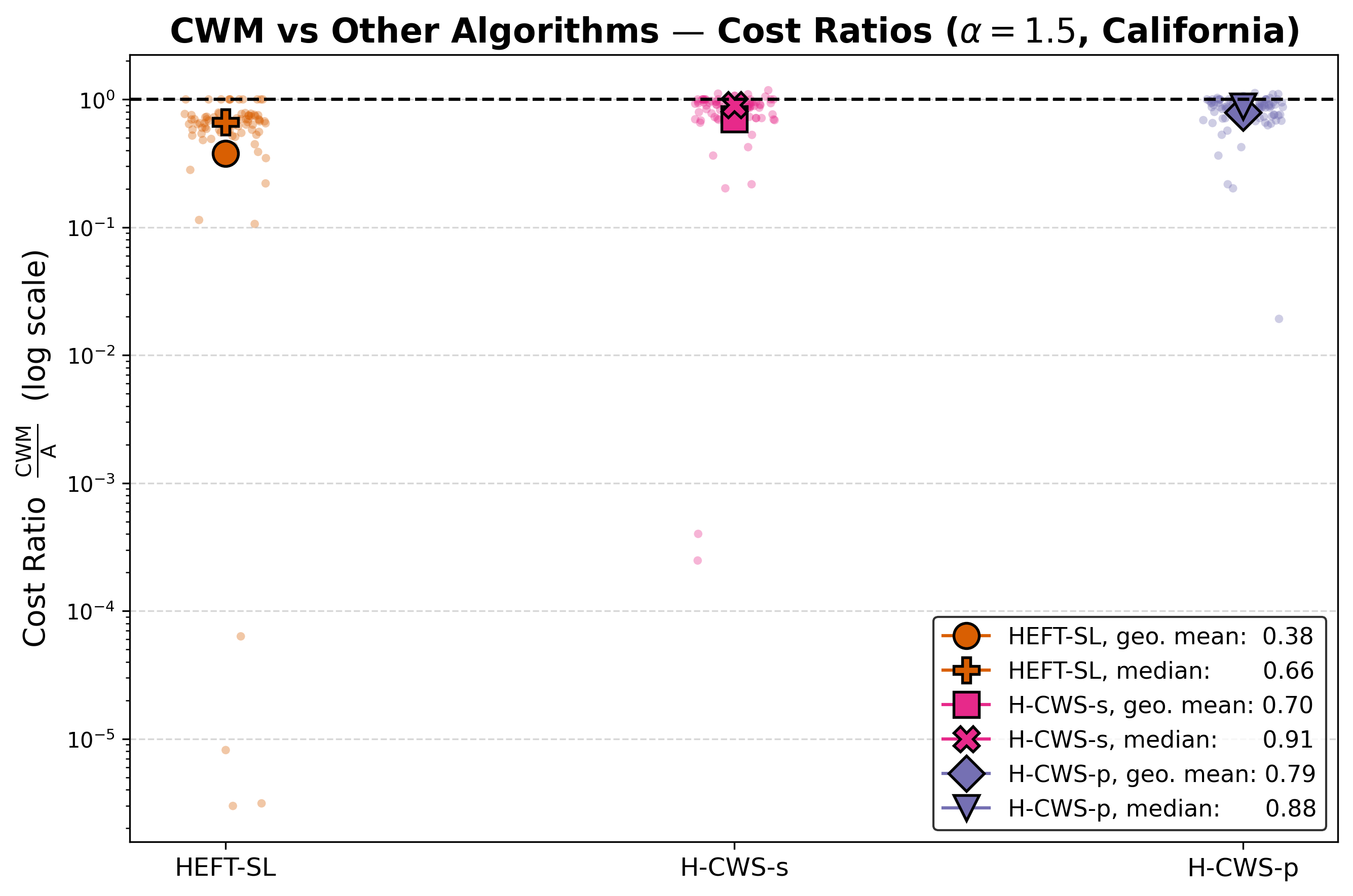}}
  \end{subfigure}
  \caption{Carbon cost ratios for each competitor for the Germany profile (left) and the California profile (right) and the deadline $D=1.5\times M$.}
  \label{fig.cost-ratios_cal_ger}
\end{figure}
\begin{figure}[h]
\centering
\begin{subfigure}{0.48\textwidth}
  \centerline{\includegraphics[width=\linewidth]{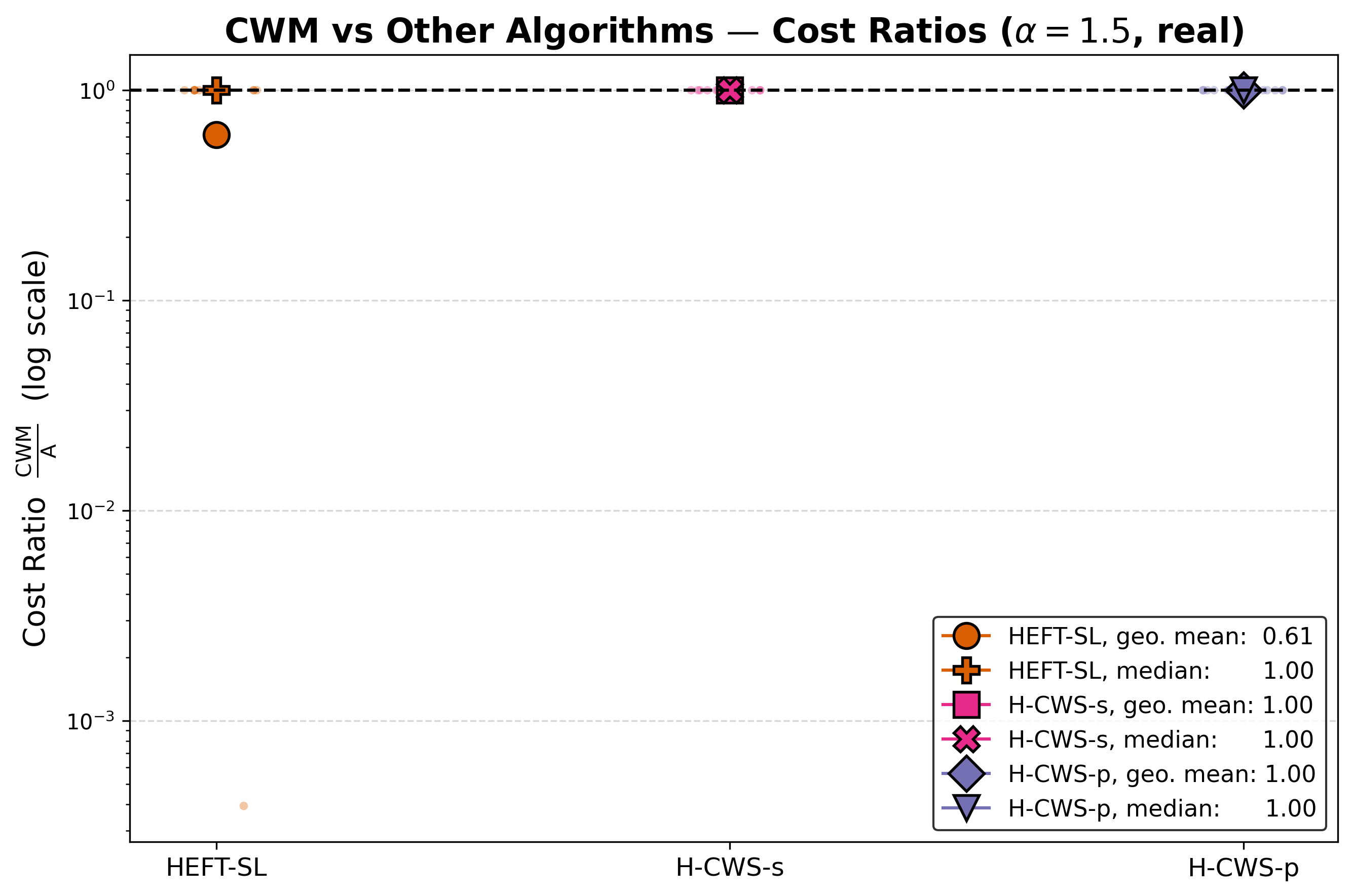}}
\end{subfigure}
\hfill
\begin{subfigure}{0.48\textwidth}
  \centerline{\includegraphics[width=\linewidth]{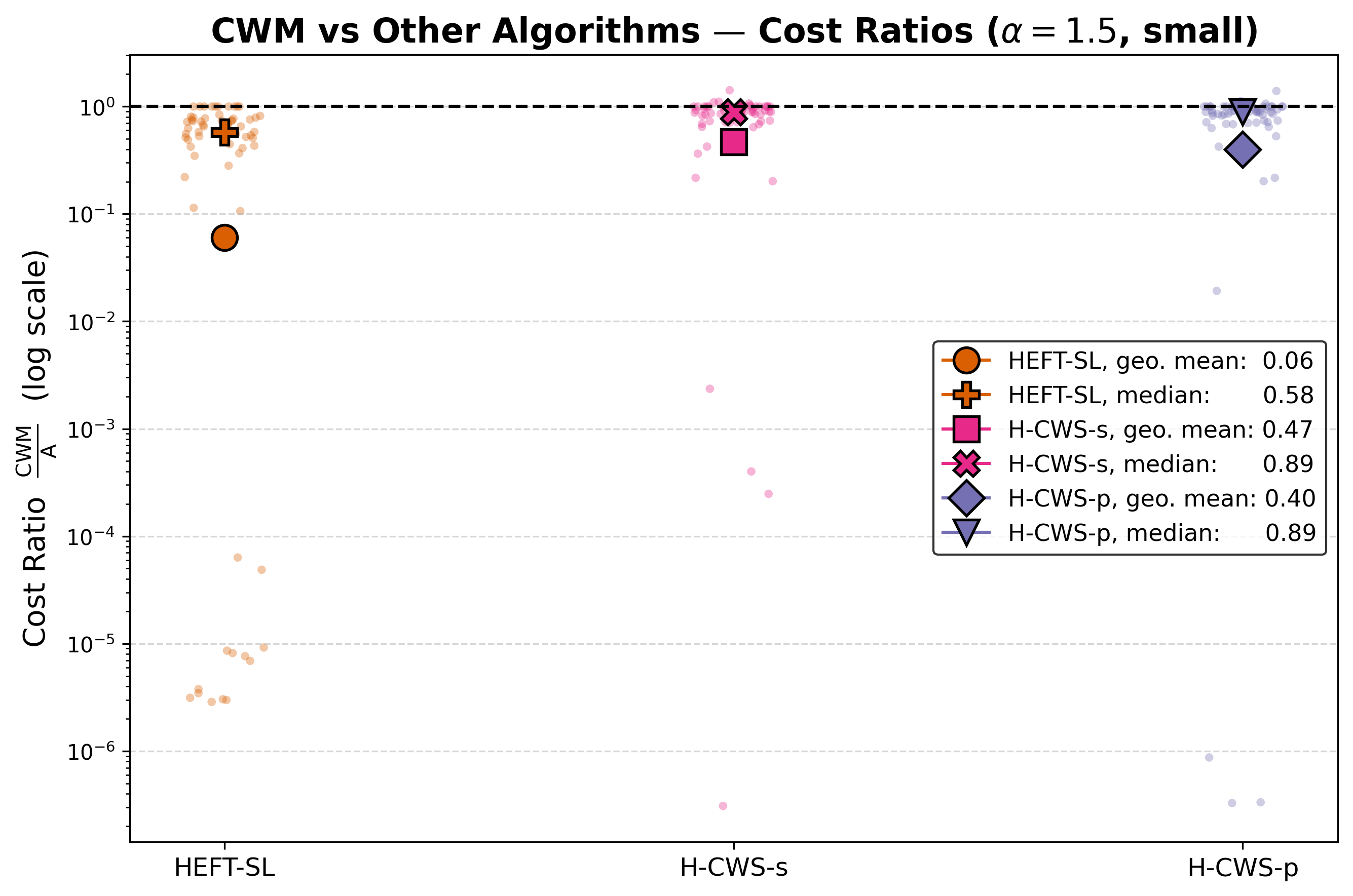}}    
\end{subfigure}
\begin{subfigure}{0.48\textwidth}
  \centerline{\includegraphics[width=\linewidth]{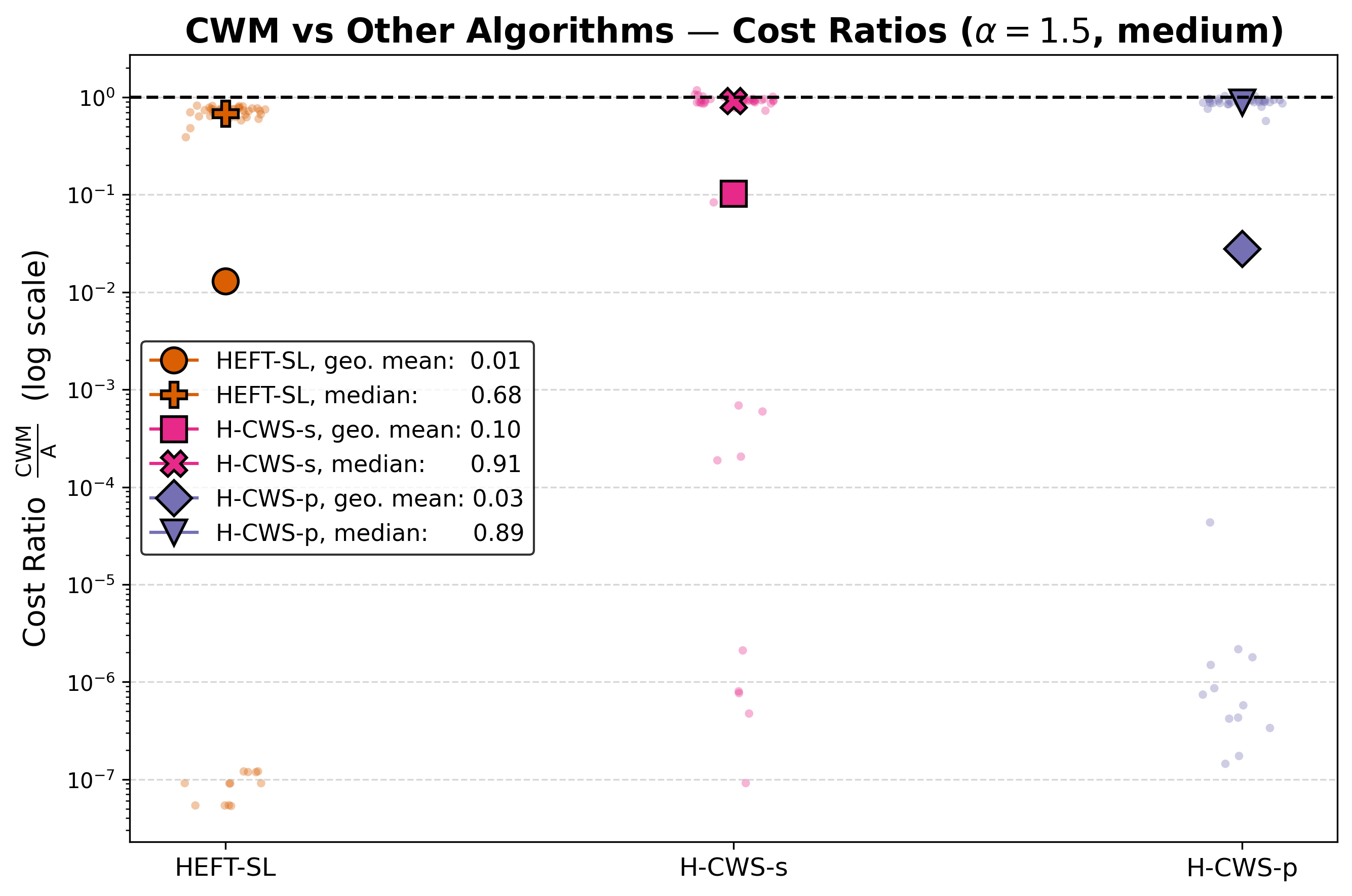}}    
\end{subfigure}
\hfill
\begin{subfigure}{0.48\textwidth}
  \centerline{\includegraphics[width=\linewidth]{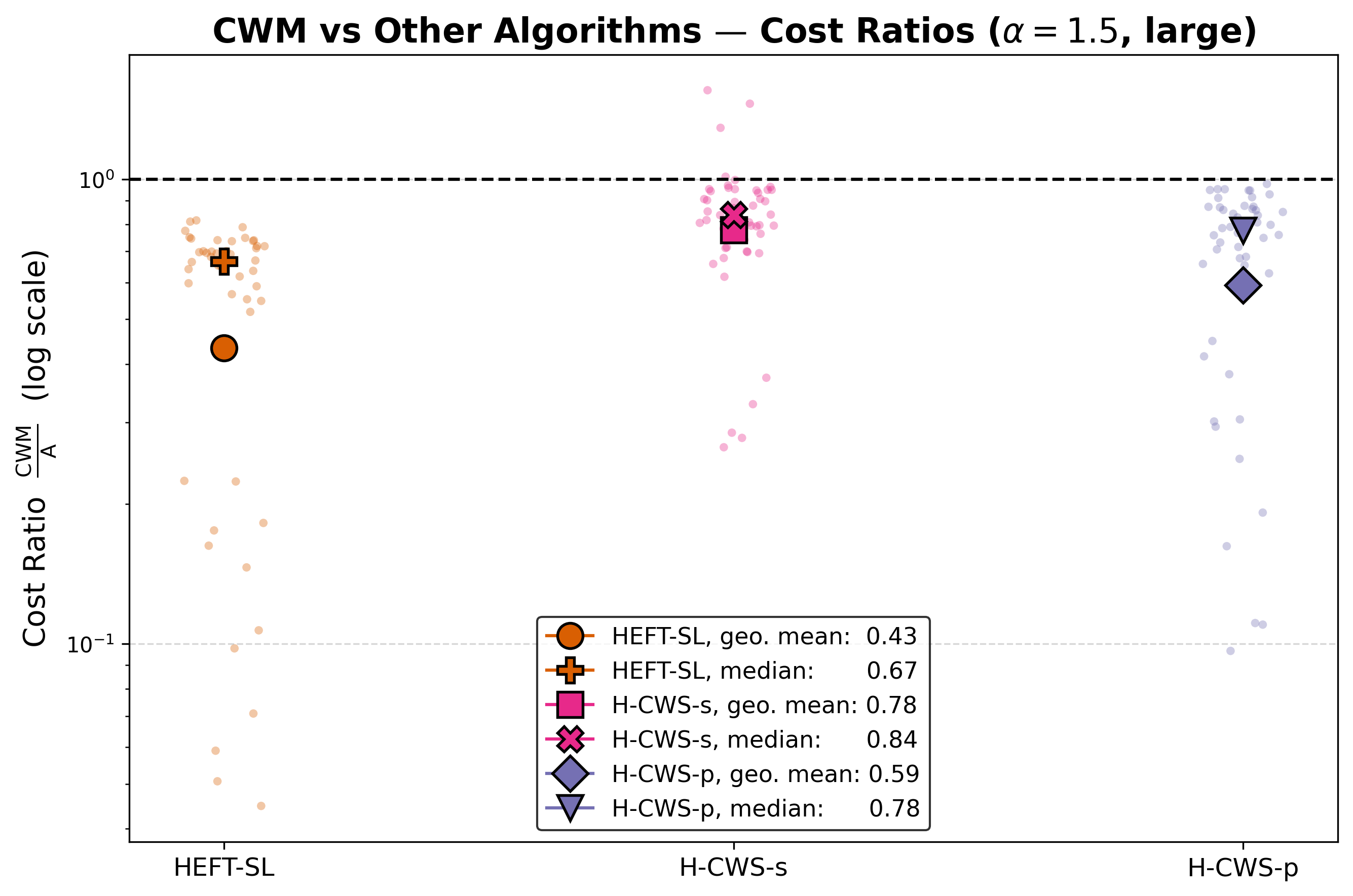}}    
\end{subfigure}
\caption{Carbon cost ratios for each competitor for tiny real-world ($\leq 60$ tasks, top left), small ($\approx 200 - 4000$ tasks, top right), medium ($\approx 8000 -15000$ tasks, bottom left) and large ($\approx 20000-30000$ tasks, bottom right) workflows 
and deadline $D=1.5\times M$.}
\label{fig.cost-ratios_size}
\end{figure}
}{}
\end{document}